\documentclass[12pt, a4paper]{article}
\usepackage{amsfonts, amssymb}
\usepackage{graphicx}
\usepackage{latexsym,amsmath,color,array}

\usepackage{natbib}
\usepackage{enumerate}
\usepackage{bbding}
\usepackage{amssymb}
\usepackage{amsmath}
\usepackage{graphicx}
\usepackage{amsmath}
\usepackage{bbm}
\usepackage{mathtools}
 \usepackage{color}
 \usepackage{array}
 \usepackage{subfigure}
 \usepackage{multirow}
 \usepackage[dvipsnames]{xcolor}
 \usepackage{makecell}
 \usepackage{colortbl}
\newcommand{\Var}{\mbox{\rm Var}}
\newcommand{\expect}{\mathbb{E}}

\def\1{\mathbb{I}}

\def\ind{\mathbbm{1}}

\newcounter{thm}[section]
\newcounter{appen}[section]
\newcounter{assum}[section]

\newtheorem{theor}[thm]{Theorem}

\newtheorem{assumption}[assum]{Assumption}

\newtheorem{lemma}[appen]{Lemma}

\newenvironment{proof}[1][Proof]{\noindent \textbf{#1.}
}{\rule{0.5em}{0.5em}}

\newcommand{\red}{\color{black}}
\newcommand{\green}{\color{black}}

\begin{document}

\title{A {\red likelihood-based} approach for cure regression models}
\author{Kevin Burke\footnote{Corresponding author. University of Limerick, Ireland; kevin.burke@ul.ie}  \hspace{3cm}
Valentin Patilea\footnote{CREST, Ensai, France; patilea@ensai.fr. \newline\newline\indent\indent
Both authors acknowledge the support of the Irish Research Council and the French Ministry of Foreign Affairs through the \emph{Ulysses} scheme. Valentin Patilea acknowledges support from the research program \emph{New Challenges for New Data} of Fondation du Risque/Institut Louis Bachelier and  LCL.} }
\date{\today}

\maketitle

\begin{abstract}

We propose a new {\red likelihood-based} approach for estimation, inference and variable selection for parametric cure regression models in time-to-event analysis under random right-censoring.  In {\red this} context,  it often happens that some subjects are ``cured'', {\red i.e., they will never} experience the event of interest. Then, the sample of censored observations is an unlabeled mixture of cured  and ``susceptible'' subjects.  Using inverse probability censoring weighting (IPCW), we propose {\red a likelihood-based estimation procedure for the cure regression model without making assumptions about the distribution of survival times for} the susceptible subjects. {\red The IPCW approach does require a preliminary estimate of the censoring distribution}, for which general parametric, semi- or non-parametric approaches {\red can} be used. The incorporation of a penalty term in our {\red estimation procedure} is straightforward; {\red in particular,} we propose $\ell_1$-type penalties for variable selection. Our theoretical results  are derived under mild assumptions. Simulation experiments and real data analysis illustrate the effectiveness of the new approach.

\smallskip

{\bf Keywords.} Binary regression; iid representation; Inverse probability censoring weighting; Penalized likelihood.

{\bf MSC 2010 subject classification:}  62N01, 62N02, 62J07
\end{abstract}

\qquad

\newpage

\section{Introduction\label{intro}}

Standard survival models assume that all individuals experience the event of interest eventually (see \citet{kalbprent:2002}). However, this  assumption is not always tenable since, for example, some diseases might only activate when specific biological and/or lifestyle traits are present, or immunity may result from successful curative treatment (or a combination of both treatment and pre-treatment attributes); individuals who will {\red never} experience the event are referred to as \emph{cured} or \emph{non-susceptible}. In practice, the binary ``cure status'', $B \in \{0, 1\}$ (where $B = 1 \Leftrightarrow$ ``{\red the individual is }cured''), typically cannot be measured directly, and, so, survival studies are required.

Let $T \in (0, \infty]$ denote the survival time, and note in particular that, in contrast to standard survival models, the support includes $T=\infty$ corresponding to cured individuals. Thus, $B = \ind(T=\infty) \sim \text{Bernoulli}(\pi)$ where $\pi$ is the cure probability. Furthermore, let $T_0 \in (0, \infty)$, be the ``latency'' time {\red (i.e., the survival time for an uncured individual) with survivor function} $S_{T_0}(t) = \Pr(T_0 > t)$. By construction,  $T=T_0$ when $B=0$, and, since $\Pr(T > t \mid B) = (1-B) S_{T_0}(t) + B$, we have that \mbox{$S_T(t) = \expect\{\Pr(T > t \mid B)\} =$} \mbox{$(1-\pi) S_{T_0}(t) + \pi$} {\red such that $\lim_{t\rightarrow\infty}S_T(t) = \pi$}. Most often, interest centres on modelling the {\red effect of} covariates, $X$, {\red on} $\pi$. Thus, covariates enter via $S_T(t\mid X) = \{1-\pi(X)\} S_{T_0}(t\mid X) + \pi(X)$
where $\pi(X) = \pi(X;\theta) \in (0,1]$ is the cure regression function, and $\theta$ is a vector of parameters; typically, although not necessarily, covariates enter through a linear predictor, i.e., $\pi(X;\theta) = \pi(X^\top\theta)$  where $X = (X_{(0)}\equiv1, X_{(1)}, \ldots, X_{(p)})^\top$ and $\theta = (\theta_{(0)},\theta_{(1)},\ldots,\theta_{(p)})^\top $ respectively. (Here and in the following, for any matrix $A$,   $A^\top$ denotes its transpose.)

Parametric cure models have been considered by various authors \citep{berkgage:1952,farewell:1977,peng:1998}, but the $\theta$ estimates can be sensitive to the choice of latency model \citep{yuetal:2004}. The semi-parametric cure model consists of a parametric cure regression model and a semi-parametric proportional hazards (PH) or accelerated failure time (AFT) latency model \citep{pengdear:2000, syetal:2000, litaylor:2002, zhangpeng:2007}; estimation requires the EM algorithm \citep{dempster:1977} in combination with (modified) partial likelihood \citep{cox:1972,cox:1975} (PH case) or rank regression \citep{ritov:1990,tsiatis:1990} (AFT case).
%While these approaches do not specify the baseline latency model,
However, structural model assumptions are %still
made in the way that covariates enter {\red the latency component} (i.e., PH or AFT structure), and estimation of the cure component can be sensitive to such choices. (See \citet{amicovankeilgom:2018} for an excellent review of cure literature.)

Our approach differs from the aforementioned. First note that, if $B$ was directly observable, one would simply apply standard binary regression (e.g., logistic or probit) without modelling $S_{T_0}$. We define a ``proxy'' or ``synthetic'' variable which we denote by $B^*$, then replace $B$ by $B^*$ and proceed with classical binary regression. We {\red construct} $B^*$ using inverse probability censoring weighting (IPCW) arguments \citep{robinsfink:2000, van2003unified}. This approach obviates the need for a latency model -- although IPCW does require estimation of the censoring distribution. At first sight, we trade one missing data framework (EM) for another (IPCW), and latency estimation for censoring estimation. However, our novel procedure has key advantages: (i) unlike existing EM approaches, which are not so easily generalized beyond PH and AFT latency models, our framework straightforwardly permits a wide range of censoring models, (ii) the censoring model can be validated using standard model checking techniques since, unlike the latency model, it is directly identifiable from the observed data, and (iii) once $B^*$ has been computed (in a single initial step rather than in an iterative EM fashion), one may avail of standard, fast (penalized) GLM estimation procedures.

The remainder of this article is organized as follows. In Section \ref{sec:prelim}, we introduce a new result which forms the basis of the proposed estimation procedure of Section \ref{sec:est}{\red, including} a penalized version for the purpose of variable selection. Asymptotic theory is provided in Section \ref{sec:ass}, along with empirical evidence via simulation in Section \ref{sec:sim}. A real data example is given in Section \ref{sec:data}, {\red and} we close with some remarks in Section \ref{sec:disc}.
The proofs are postponed to the Appendix (Section \ref{sec:app}).

\section{Preliminaries\label{sec:prelim}}

Before we proceed, we define the censoring time $C \in (0, \infty)$ with survivor function $S_C(t)$, whose support excludes infinity since, practically, observation windows are finite. Furthermore, let $Y = T \wedge C$ be the observed time ($\wedge$ is the minimum operator) and $\Delta = \ind(T \le C) = (1-B)\ind(T_0 \le C)$ be the event indicator. %Neither $T$ nor $B$ are fully observed; hence, inference is made through $Y$ and $\Delta$.
We will assume the following:
\begin{align}
T_0 &\perp C  \mid X, \label{ass1} \\[0.2cm]
B &\perp (T_0, C)   \mid X \label{ass2},
\end{align}
where (\ref{ass1}) is the standard independence assumption made throughout survival literature, and (\ref{ass2}) is introduced in the cure context to ensure that the cure regression model is identifiable. In particular, assumptions (\ref{ass1}) and (\ref{ass2}) guarantee that $T \perp C \mid X$ (see the Appendix, Lemma \ref{indep_eq}). With these assumptions in place, it can then be shown that
\begin{equation}
\expect\left(\left.\frac{\Delta}{S_C(Y-\mid X)} ~\right| X\right) ~=~ \expect(1-B\mid X) ~=~ 1 - \pi(X) \label{expectcure}
\end{equation}
where $S_C(t-\mid X) = \Pr(C \ge t \mid X).$ This result
is the core of our estimation scheme which is described in Section \ref{sec:est}.
In fact, (\ref{expectcure}) is a special case of the more general result
\begin{equation}
\expect\left(\left.\frac{\Delta r(Y,X)}{S_C(Y-\mid X)} ~\right| X\right) = \expect\{r(T_0,X)\mid X\} \{1-\pi(X)\} \label{expectcure2}
\end{equation}
but with $r(\cdot,\cdot)=1$; indeed, it is (\ref{expectcure2}) {\red which} is proved in Lemma \ref{core_id}.
It is worth highlighting that (\ref{expectcure2}) is an application of the Inverse-Probability Censoring Weighting (IPCW) approach \citep{robinsfink:2000} extended to the case where a cured proportion exists. %; (\ref{expectcure2}) reduces to the usual IPCW approach when $\pi(X)\equiv0$.

%
%In our applications, we will make use of a logistic regression function, $\pi(X) = 1/\{1 + \exp(- \theta^\top X)\}$, but one may, of course, use alternative parametric forms (e.g., probit or complementary log-log, or forms in which covariates enter in non-linear ways). Note that we have not suggested any particular model for $S_{T_0}(t\mid X)$ as, in our proposed estimation procedure, this function is completely unspecified.

\section{Estimation and inference\label{sec:est}}

If independent and identically distributed (iid) replicates of $(B_i,X_i)$ were observed, we would use the Bernoulli log-likelihood,
\begin{equation}\label{likeB}
\ell(\theta) = \sum_{i=1}^n \left[B_i \log \pi_i (\theta)  + (1-B_i) \log(1-  \pi_i (\theta) )\right],
\end{equation}
 where  $\pi_i (\theta)= \pi(X_i;\theta)$.
Of course, (\ref{likeB}) is not operational {\red in the current context since} $B_i$ is unobserved, but it serves as our motivation for estimation and inference using iid replicates $(Y_i,\Delta_i,X_i)$. %, $1\leq i \leq n$.

\subsection{{\red Likelihood-based} estimation}%\label{sec:est}}

We now define
\begin{equation}\label{def_Bistar}
B_i^*(S_C) = B_i^*(Y_i, \Delta_i, X_i, S_C) = 1 - \frac{\Delta_i}{S_C(Y_i-\mid X_i)},
\end{equation}
{\red such that} $\expect\{B_i^*(S_C)\mid X_i\} = \expect(B_i\mid X_i)$ {\red follows} from (\ref{expectcure}). Assuming initially that $S_C$ is known, we propose replacing $B_i$ with $B_i^*(S_C)$ in (\ref{likeB}) to obtain
\begin{equation}
\ell^*(\theta) = \sum_{i=1}^n\left[B_i^*(S_C) \log \pi_i (\theta)+ \{1- B_i^*(S_C)\} \log(1-\pi_i (\theta) )\right]   \omega_i.  \label{likedata}
\end{equation}
Unlike $\ell(\theta)$ however, $\ell^*(\theta)$ is formed using the observable quantities $Y_i$ and $\Delta_i$ rather than the unobservable $B_i$. Here, $\omega_i=\omega(X_i)$ are positive weights which are introduced as a technical device when deriving general asymptotic results (but we anticipate that $\omega_i  \equiv 1$ in practically all applications). {\red In the sense that $\ell^*(\theta)$ is based on $\ell(\theta)$, we refer to our procedure as ``likelihood-based''. While the phrase ``pseudo-likelihood'' has been used in similar contexts to ours \citep{Xieliu:2005}, we prefer to avoid this terminology due to its long-standing usage in settings where incorrect models can still yield valid inferences \citep{Gourierouxetal:1984}; our setting is rather different to this in that the model remains unchanged, but the unobservable $B_i$ is replaced with the observable surrogate $B_i^*$.} The score function,
$U^*(\theta)= \partial \ell^*(\theta)  / \partial \theta = \sum_{i=1}^n U^*_i(\theta)$ where
\begin{equation}\label{score_hitomi}
U_i^*(\theta) = \frac{\left[B_i ^*(S_C) - \pi_i (\theta)\right]\omega_i}{\pi_i (\theta)  (1- \pi_i (\theta) )}\,  \frac{\partial  \pi_i (\theta)  }{\partial \theta }\,  \in\mathbb{R}^{p+1},
\end{equation}
 is  unbiased due to property (\ref{expectcure}), {\red and}, in the case of a logistic cure regression model, $  U_i^*(\theta) = [B_i^*(S_C)-  \pi_i (\theta) ] X_i \omega_i  $ with $ \pi_i (\theta) = 1/\{1 + \exp(- \theta^\top X_i)\} $. Furthermore, using standard inequalities, we show in the Appendix (Lemma \ref{lik_ineq_proof}) that, when the cure regression model is identifiable, $\expect\{\ell^*(\theta)\} < \expect\{\ell^*(\theta_0)\}$ $\forall\,\theta \ne \theta_0,$
where $\theta_0$ is the true parameter vector. %(see the. %See Lemma \ref{lik_ineq_proof} in the Appendix.
%Hence, $\ell^*(\theta)$ is a legitimate criterion for estimation and inference on $\theta$.

Note that we have explicitly written $B^*_i(S_C)$ as a function of $S_C$ (while its dependence on $Y_i$, $\Delta_i$, and $X_i$ is implicit) since, typically, $S_C$ must be estimated in practice, i.e., we will use $B_i^*(\widehat S_C)$; this is standard in IPCW applications (and recall that (\ref{expectcure}) is based on IPCW arguments). Furthermore, the asymptotic theory of Section \ref{sec:ass} requires only that the estimator for $S_C$ has an iid representation.
Thus, while $S_{T_0}$ is completely unspecified in our proposal, $S_C$ is estimated, and this can essentially be done in an arbitrarily flexible way.
Therefore, for practical purposes, we propose the cure estimator defined as
\begin{equation}\label{mle}
\widehat \theta = (\widehat\theta_{(0)},\widehat\theta_{(1)},\ldots,\widehat\theta_{(p)})^\top = \arg\max_{\theta} \widehat \ell^*(\theta)
\end{equation}
where
\begin{equation}
\widehat \ell^*(\theta) = \sum_{i=1}^n\left[B_i^*(\widehat S_C) \log \pi_i (\theta) + \{1- B_i^*(\widehat S_C)\} \log(1-\pi_i (\theta) )\right]  \omega_i \label{likedata_b}.
\end{equation}
The weights $\omega_i = \omega(X_i)$ serve in theory to control the behavior of general estimates of $B_i^*(S_C)$ in regions of low covariate density. Since (\ref{likedata_b}) is the usual Bernoulli log-likelihood with $B_i$ replaced by $\widehat B_i^* = B_i^*(\widehat S_C)$, standard GLM estimation procedures can be used, and the EM algorithm is avoided. {\red Interestingly, as discussed further in Section \ref{sec:normal}, it appears in our case that replacement of $B_i$ with $\widehat B^*_i$ produces more efficient cure parameter estimates than with $B^*_i$. (Of course, $B_i$ itself would produce the most efficient estimates if it were available.)} Note that the $\widehat B^*_i$ variables play a somewhat similar role to ``synthetic observations'' as used by \citet{kouletal:1981} in a different context (see also \citet{delecroixetal:2008}). % in least squares estimation for censored survival data. However, their response variable is a survival time, whereas we have a binary cure indicator.

%Note that current estimation procedures for semi-parametric cure models also involve replacing $B_i$ in (\ref{likeB}) with an expected value, but in an iterative EM fashion \citep{pengdear:2000,syetal:2000,litaylor:2002, zhangpeng:2007}, whereas $B_i^*(\widehat S_C)$ is computed in one step followed by maximization of (\ref{likedata_b}).  The reason for this difference comes from the modeling approach. In existing approaches, models are assumed for both the conditional survival time of the uncured individuals,  $S_{T_0}(t\mid X)$, and the cure proportion, $\pi(X)$, that together completely determine $S_T(t\mid X)$. Since $S_T(t\mid X)$ is identifiable from the observed data, the EM iterations between $S_{T_0}(t\mid X)$ and $\pi(X)$ represent a natural way to estimate these assumed model components. In contrast, our approach uses the fact that  $S_C(t\mid X)$ is also directly identifiable, and, thus, we neither  need to impose assumptions on $S_{T_0}(t\mid X)$ nor use an iterative procedure.

\subsection{Variable selection\label{sec:varselect}}

With the above in place, by construction, it is also straightforward to include a penalty for the purpose of variable selection. In particular, we consider the alasso (adaptive least absolute shrinkage and selection operator) penalty of \citet{zou:2006}. Although, so far, we have not specified the functional form of $\pi(X;\theta)$, here, we will assume that $\pi(X;\theta) = \pi(X^\top \theta)$, and, thus, $X\in \mathbb{R}^{p+1}$ with $X_{(0)}\equiv 1$.  Then, the alasso estimator is
\begin{equation}\label{alasso_est}
\widehat \theta_\lambda  =(\widehat \theta_{\lambda ,(0)},\widehat \theta_{\lambda ,(1)},\ldots,\widehat \theta_{\lambda ,(p)})^\top = \arg\max_{\theta\in\Theta } \widehat \ell^*_{\lambda}(\theta)
\end{equation}
where
$
\widehat \ell^*_{\lambda}(\theta) = \widehat \ell^*(\theta) -  \lambda \sum_{j=1}^p w_j |\theta_{(j)}| %\label{alasso_lik_def}
$
%is the penalized likelihood function, $\widehat \ell^*(\theta)$ is defined in (\ref{likedata_b})
 with tuning parameter $\lambda \ge 0$ and (potentially adaptive) weights $w_j \ge0$ for $1 \le j \le p$.
Here, as is usual, the intercept, $\theta_{(0)}$, is not penalized in (\ref{alasso_est}), and, furthermore, typically, the covariates {\red are} standardized. Setting $w_j = 1$  $\forall j$ yields the lasso penalty, which penalizes all coefficients equally \citep{tibshirani:1996}, while $w_j = 1/|\widehat\theta_{(j)}^{(0)}|^\gamma$ for some $\gamma >0$ yields the alasso penalty (we will set $\gamma=1$ as is most common in practice). In the latter case,  $\widehat\theta_{(j)}^{(0)}$ may be any consistent estimator of $\theta_{(j)}$, and, typically,  $\widehat\theta_{(j)}^{(0)}=\widehat\theta_{(j)}$, where $\widehat\theta_{(j)}$ is the $j$th unpenalized estimator from (\ref{mle}). Details on implementation aspects of the alasso in our context (optimization procedure and tuning parameter selection) can be found in the Appendix (Section \ref{sec:imp}).

\section{Asymptotic results\label{sec:ass}}

Our asymptotic results are proved under some minimal moment assumptions on the observed variables completed by  some mild high-level assumptions on the cure regression model and on the model for the censoring variable. These conditions are quite natural in the context of right-censored data when  covariates are present and are to be verified on a case by case basis according to the context of the application. In this section we use the  notation $\pi(\theta ) = \pi (X;\theta),$
$\partial \pi(\theta)/ \partial \theta= \partial \pi(X; \theta)/ \partial \theta$,
and $B^*_i(S_C)$ (hence, $B^*_i(\widehat S_C)$) as defined in equation (\ref{def_Bistar}).

\begin{assumption}\label{dgp} \emph{(The data)}
The observations  $(Y_i,\Delta_i,X_i),$  $1\leq i \leq n,$ are independent {\red replicates} of $(Y,\Delta,X)\in\mathbb{R}\times \{0, 1\} \times \mathcal{X}$, where $\mathcal{X}$ is some covariate space.  Moreover, $\mathbb{E}[\Delta/S_C (Y-\mid X)]<\infty.$
\end{assumption}

Let $\mathcal{M}=\{\pi(\theta ): \theta \in\Theta \subset \mathbb{R}^{p+1}\}$ be a generic parametric cure regression model, that is a set of functions of the covariate vector $X$ indexed by $\theta$ in some parameter space $\Theta$, {\red e.g.,} the logistic cure model is $\pi(\theta) = 1/\{1 + \exp(- \theta^\top X)\}.$

\subsection{Consistency}
For the definition of Glivenko-Cantelli function classes, we refer  to \citet{van2000asymptotic}.

\begin{assumption}\label{reg_cure_ass}
\emph{(The cure regression model)}
\begin{enumerate}
\item\label{ass_ome} The  weight $\omega(X)$ is bounded, almost surely nonnegative and has a positive expectation.

\item\label{ass_truth} There exists $\theta_0\in\Theta$ such that $\Pr(T=\infty  \mid X) = \pi(\theta_0)$. Moreover, there exists $0<c <1/2$ such that, for any $\theta\in\Theta$, $\Pr(c\leq  \pi(\theta) \leq 1-c) = 1.$

\item\label{ass_sep} For any $\varepsilon >0,$  $\inf_{\|\theta-\theta_0\| >\varepsilon } \mathbb{E}[|\pi(\theta)- \pi(\theta_0)|\omega (X)] >0$.

\item\label{ass_GC} The model $ \mathcal{M}$ is a $\mathbb{P}_{X}-$Glivenko-Cantelli class of functions of $X$ with constant envelope.
\end{enumerate}

\end{assumption}

\begin{assumption}\label{ulln_ass}
\emph{(Uniform law of large numbers)}
The estimator  $\widehat S_C (\cdot\mid  \cdot)$ satisfies the law of large numbers uniformly over the class of the logit transformations of the functions in $\mathcal M$:
\begin{equation*}%\label{cons1}
\sup_{\theta \in\Theta} \left|  \frac{1}{n} \sum_{i=1}^{n} \left[ B^*_i(\widehat S_C) - B^*_i( S_C) \right] \omega(X_i) \log\left( \frac{\pi_i(\theta)}{1-\pi_i(\theta)} \right) \right| =o_{\mathbb{P}}(1).
\end{equation*}
\end{assumption}

\quad

For simplicity, we consider a bounded weight $\omega(X)$ and assume that {\red $\pi_i(\theta)$ stays} uniformly away from 0 and 1 {\red $\forall i$}. {\red Note that Assumption \ref{reg_cure_ass}.\ref{ass_sep} provides an indentifiability condition which guarantees} that $\theta_0$ is a well-separated maximum of {\red $\expect\{\ell^*(\theta)\}$.} {\red It is satisfied, for example, by} logistic or probit regression {\red models} when the covariates are not redundant, i.e., when $\mathbb{E}[XX^\top\omega (X)]$ is an invertible matrix.

Let us provide some mild sufficient conditions implying the uniform convergence of Assumption \ref{ulln_ass}.
These sufficient conditions involve a threshold {\red, commonly used in cure literature, which} is typically justified as representing a total follow-up of the study. {\red Usually, this is assumed} to be independent of the covariates, {\red but} we allow it to depend on the covariates in an arbitrary way.

\begin{lemma}\label{suff_cdt_ulln}
Assume that there exists $\tau(x)$ such that, for any $x$, $\Pr(T_0 > \tau(x)) = 0 $ and $\inf_{x \in\mathcal{X}, \;\omega(x)>0} S_C(\tau(x)-\mid x) >0. $ Moreover,  Assumptions {\red\ref{reg_cure_ass}.\ref{ass_truth} and \ref{reg_cure_ass}.\ref{ass_GC}} hold true.  If
\begin{equation}\label{suff_c1}
\sup_{x \in\mathcal{X}, \;\omega(x)>0} \;\sup_{y\leq\tau(x)}\left|\widehat S_C(y-\mid x) - S_C(y-\mid x)  \right| = o_{\mathbb{P}}(1) ,
\end{equation}
then the uniform convergence  in Assumption \ref{ulln_ass} holds true.
\end{lemma}

\quad

The common parametric, semiparametric and nonparametric estimators $\widehat S_C$ satisfy condition (\ref{suff_c1}). Several examples are provided in the monographs by \citet{borgan1993statistical} and \citet{kalbprent:2002}, and, for convenience, some examples are  recalled in the Appendix (Section \ref{sec:unifiid}). The consistency of our cure estimator is stated in the following result.

\begin{theor}\label{conv_th1}
Let Assumptions \ref{dgp}, \ref{reg_cure_ass} and \ref{ulln_ass} hold true.
%If, in addition, either Assumption  \ref{ulln_ass} or condition (\ref{suff_c1}) holds true,
Then $\widehat \theta -\theta_0=o_{\mathbb{P}}(1)$.
\end{theor}

\subsection{Asymptotic normality\label{sec:normal}}

\begin{assumption}\label{reg_cure_ass_2}
\emph{(The cure regression model)}

\begin{enumerate}
\item  For any $x\in\mathcal{X}$, the map $\theta\mapsto \pi (x; \theta)$ is twice continuously differentiable.

\item\label{A_pos_def} The true value $\theta_0$ is an interior point of $\Theta$, $$\mathbb{E}\left[\left\| \frac{\partial \pi (\theta_0)}{\partial \theta}  \right\|^2 \right]<\infty$$ and the $(p+1)\times (p+1)-$matrix
$$
A(\theta_0)=\mathbb{E}\left[ \frac{\omega(X)}{\pi (\theta_0)[1-\pi (\theta_0)]} \frac{\partial \pi (\theta_0)}{\partial \theta}    \frac{\partial \pi (\theta_0)}{\partial \theta}  ^\top\; \right]
$$
is positive definite.

\item\label{ULLN_as_nor} For any $0\leq k\leq l \leq  p,$ the families of functions of $x$ indexed by $\theta$
%$$
%\mathcal{F}_{1,kl}\!= \!\left\{ \frac{\partial^2 \pi}{\partial \theta_{(k)} \partial \theta_{(l)} }(x;\theta)   : x\!\in\mathcal{X}, \theta\! \in \Theta \! \right\}\! , \quad \! \!\mathcal{F}_{2,kl}\!=  \!\left\{\!\! \left(\! \frac{\partial \pi}{\partial \theta_{(k)} } \; \frac{\partial \pi}{\partial \theta_{(l)}} \! \right) \!\! (x;\theta)   : x\!\in\mathcal{X}, \theta\! \in \Theta \! \right\}
%$$
\begin{align*}
\mathcal{F}_{1,kl}\! &= \!\left\{ \frac{\partial^2 \pi}{\partial \theta_{(k)} \partial \theta_{(l)} }(x;\theta)   : x\!\in\mathcal{X}, \theta\! \in \Theta \! \right\}\! ,\\ %\quad \! \!
\mathcal{F}_{2,kl}\! &=  \!\left\{\!\! \left(\! \frac{\partial \pi}{\partial \theta_{(k)} } \; \frac{\partial \pi}{\partial \theta_{(l)}} \! \right) \!\! (x;\theta)   : x\!\in\mathcal{X}, \theta\! \in \Theta \! \right\},
\end{align*}
are $\mathbb{P}_{X}-$Glivenko-Cantelli classes of functions of $X$ with integrable envelopes.

\end{enumerate}

\end{assumption}

\begin{assumption}\label{uclt_ass2}
\emph{(I.I.D. representation)}
Let $\varphi(X)$ be a  vector-valued function  such that $\mathbb{E}\{\| \varphi(X)\|^2\}<\infty$. Then  there exists $\mu_C^\varphi (Y,\Delta, X)$ a zero-mean vector-valued function  that depends on $\varphi(X)$,  such that $\mathbb{E}\{\| \mu_C^\varphi (Y,\Delta, X)\|^2\}<\infty$ and
$$
\frac{1}{n} \sum_{1\leq i\leq n} \left[ B^*_i(\widehat S_C) - B^*_i( S_C)  \right]\varphi (X_i) = \frac{1}{n} \sum_{1\leq i\leq n} \mu_C^\varphi(Y_i,\Delta_i, X_i) + o_{\mathbb{P}}(n^{-1/2}) ;
$$

\end{assumption}

\quad

Assumption \ref{reg_cure_ass_2} introduces mild standard regularity conditions on the cure regression model. In particular, Assumption {\red \ref{reg_cure_ass_2}.\ref{ULLN_as_nor} yields} the uniform law of large numbers and guarantees that the remainder terms in the standard Taylor expansion used to {\red establish} asymptotic normality are uniformly negligible. Such an assumption on the complexity of the classes of first and second order derivatives of the functions in the model are satisfied by the standard parametric models such as logit and probit models. As an alternative to Assumption {\red\ref{reg_cure_ass_2}.\ref{ULLN_as_nor}}, we could impose condition (\ref{suff_c1}) and slightly stronger regularity conditions on the model $\mathcal{M}$. The details are \color{black} provided in \color{black} the proof of Theorem \ref{prop_tcl}. Furthermore, note that an asymptotic representation as required in Assumption \ref{uclt_ass2} is very common in survival analysis models with `nuisance' parameters which {\red may} belong to a space of functions. {\red In the Appendix (Section \ref{sec:unifiid}), } we provide details for many standard survival estimators: Kaplan-Meier, conditional Kaplan-Meier (Beran), Cox model, transformation model, and proportional odds model.

\color{black}

In general, the expression of the function $\mu_C^\varphi(Y,\Delta, X) $  in Assumption \ref{uclt_ass2} depends on the joint law of %the observations
$(Y,\Delta, X)$. {\red Furthermore, } this function contributes to the asymptotic variance of {\red our estimator $\widehat\theta$ proposed in (\ref{mle}) which}, hence, will {\red differ from that} %from the asymptotic variance
of the infeasible {\red estimator} defined with $B^*(S_C)$ instead of $B^*(\widehat S_C)$.

\begin{theor}\label{prop_tcl}
Assume the conditions of Assumptions \ref{dgp} and \ref{reg_cure_ass}, and Lemma \ref{suff_cdt_ulln} are met. Moreover, let Assumptions \ref{reg_cure_ass_2} and \ref{uclt_ass2} hold true, and
$$
\varphi(X)= \frac{\omega(X)}{ \pi (\theta_0) [1 - \pi (\theta_0)  ]  }  \;  \frac{\partial \pi(\theta_0) }{\partial \theta}   .
$$
Then
$$
\widehat \theta -\theta_0=  A(\theta_0) ^{-1}  \frac{1}{n} \sum_{i=1}^n  \left\{ \mu(Y_i,\Delta_i, X_i;\theta_0) + \mu_C^\varphi (Y_i,\Delta_i, X_i;\theta_0) \right\}  \\ + o_{\mathbb{P}}(n^{-1/2}).
$$
where
$$
\mu(Y,\Delta, X;\theta_0)  =  \frac{ \left[B^*( S_C)  -\pi (\theta_0)\right]\omega(X)}{\pi (\theta_0)[1-\pi (\theta_0)]} \; \frac{\partial \pi(\theta_0) }{\partial \theta}
$$
and $\mu_C^\varphi (Y,\delta, X)$ is the zero-mean vector-valued function from  Assumption \ref{uclt_ass2}. In addition,
$$
\sqrt{n}\left( \widehat \theta -\theta_0 \right) \rightsquigarrow N_{p+1}\left(0, A(\theta_0) ^{-1}V(\theta_0) A(\theta_0) ^{-1}\right)
$$
with $V(\theta_0)= \Var\left\{\mu(Y,\Delta, X ;\theta_0) + \mu_C^\varphi(Y,\Delta, X ;\theta_0)\right\}$. ($\rightsquigarrow$ denotes convergence in law.)
\end{theor}

\quad

\color{black}

Following {\red an anonymous reviewer's} suggestion, let us {\red analyse} the variance $V(\theta_0)$.
It is {\red straightforward to show} that $Var\left\{\mu(Y,\Delta, X ;\theta_0)\right\}$, the variance {\red obtained} by maximizing \eqref{likedata} {\red which uses $B^*_i(S_C)$}, is larger than {\red that obtained by maximizing by \eqref{likeB} which uses $B_i$. (Of course, neither of these are feasible since  $B_i$ and $S_C$ are unknown.) However,} what {\red seems to be much less well known} is that, in many cases, the contribution of $ \mu_C^\varphi(Y,\Delta, X ;\theta_0)$, {\red due to estimating $S_C$, has the effect of making} $V(\theta_0)$ smaller than  $Var\left\{\mu(Y,\Delta, X ;\theta_0)\right\}$. This perhaps unexpected effect,  explained by  \citet{hitomi_nishiyama_okui_2008}, often occurs in semiparametric estimation {\red problems as considered herein. More specifically, their} Theorem 3 %of \citet{hitomi_nishiyama_okui_2008}, the variance of the estimator of $\theta_0$ obtained with $\widehat S_C$ will not be larger than %the variance of the (unfeasible) estimator of $\theta_0$
{\green shows that that $Var\left\{\mu(Y,\Delta, X ;\theta_0)\right\} - V(\theta_0) $ is positive semi-definite  provided that }
%obtained with  $ S_C$  {\red provided that}
$\mu_C^\varphi(Y,\Delta, X ;\theta_0)$ belongs to the tangent space $\mathcal T$ with respect to the `nuisance' parameter $S_C$. The space $\mathcal T$ is
defined as the mean-square closure of all linear combinations of scores for smooth parametric submodels passing through $S_C$. When there is no restriction on the form of  $S_C$, {\red this} tangent space is typically large and {\red very likely to} include $ \mu_C^\varphi(Y,\Delta, X ;\theta_0)$; {\red from} Remark 2 of \citet{hitomi_nishiyama_okui_2008}, {\red this is also true for parametric $S_C$ models estimated using maximum likelihood.} %the same will happen when $S_C$ belongs to a parametric model and is estimated by maximum likelihood.
Finally, \citet{hitomi_nishiyama_okui_2008} {\red show (again in their Theorem 3) that} $V(\theta_0)$ is {\red \emph{strictly} less} than  $Var\left\{\mu(Y,\Delta, X ;\theta_0)\right\}$ {\red if the} projection of the score {\red function} \eqref{score_hitomi} onto $\mathcal T$ is {\red non-null}. We claim that this is indeed the case {\red for many possible $S_C$ models, and a formal investigation of this will be a focus of our future work. However, we provide numerical experiments which support this claim in Section \ref{sec:sim}.}

\color{black}

{\red To estimate the variance of $\widehat\theta$, we require estimates of $V(\theta_0)$ and $A(\theta_0)$, respectively.} If estimates of the vectors
$ \mu(Y_i,\Delta_i, X_i;\theta)$ and $ \mu_C^\varphi (Y_i,\Delta_i, X_i;\theta)$  are available, say, $ \widehat \mu_i ( \theta)$ and $ \widehat \mu_{C,i}^\varphi (\theta )$, then $V(\theta_0) $ {\red can} be estimated by {\red the} sample covariance, $n^{-1} \sum_{i=1}^n [\widehat \mu_i (\widehat \theta) +  \widehat \mu_{C,i}^\varphi (\widehat \theta )]^{\bigotimes2}$ where $a^{\bigotimes2} = a a^T$. Meanwhile, $A( \theta_0) $ {\red can} also be estimated by standard methods.  However, the estimates $ \widehat \mu_i ( \theta)$ and $ \widehat \mu_{C,i}^\varphi (\theta )$ are often quite intractable. {\red Thus, one may alternatively use} the nonparametric bootstrap; indeed, this approach works well empirically (see Appendix, Section \ref{sec:simadd}) and is used in our real data analysis.

\subsection{Oracle properties for the adaptive lasso}

Next, we
prove consistency in variable selection for the adaptive lasso proposed in Section \ref{sec:varselect}. Moreover, we prove the asymptotic normality for the true subset of coefficients. Hence, we extend the Theorem 4 of \citep{zou:2006} to the cure regression context.
Let $\theta_0 = (\theta_{0,(0)},\theta_{0,(1)},\ldots,\theta_{0,(p)})^\top$ be the true cure regression parameter vector. Assume the true model has a sparse representation, {\red where} $\mathcal{A} = \left\{j : 1\leq j\leq p , \theta_{0,(j)}\neq 0  \right\}\cup\{0\}$. Without loss of generality, suppose $\mathcal{A} = \{0,1,\ldots,p_0\}$, $p_0<p. $ Below, the subscript $\mathcal{A}$ is used to define the subvectors or blocks in matrices with components corresponding to the indices in the set $\mathcal{A}$, i.e.,  $\theta_{\mathcal{A},0}$ is the subvector of the first $p_0+1$ components of $\theta_0 $, $\partial \pi(\theta_0)/\partial \theta_{\mathcal{A}}$ denotes the vector of %$(p_0+1)$
partial derivatives with respect to the first $p_0 + 1$ components of $\theta$, and $A_{\mathcal{A}}(\theta_{0})$ is the upper-left block of dimension $(p_0+1)\times (p_0+1)$ of the % $(p+1)\times (p+1)-$
matrix $A(\theta_0)$ defined in Assumption {\red \ref{reg_cure_ass_2}.\ref{A_pos_def}}.

\begin{theor}\label{oracle}
Assume the conditions of Theorem \ref{prop_tcl} are met and $\pi(X;\theta)$ is a given function of $X^\top \theta$. Let $\widehat \theta_\lambda$ be the estimator defined in (\ref{alasso_est}) with  $w_j = |\widehat\theta_{(j)}|^{-\gamma}$ for $\gamma >0$.
Moreover, assume that $\lambda/\sqrt{n}\rightarrow 0$ and $\lambda n ^{(\gamma-1)/2}\rightarrow \infty$. Let
$\mathcal{A}_n = \left\{j : 1\leq j\leq p , \widehat \theta_{\lambda ,(j)}\neq 0  \right\}\cup\{0\}.$
Then
\begin{enumerate}
\item $\lim_{n\rightarrow \infty} \mathbb{P}(\mathcal{A}_n = \mathcal{A} )=1$.\\[-0.2cm]
\item ~\\[-1.15cm]
\begin{align*}
%$%{\displaystyle
 &\widehat \theta_{\!\mathcal{A},\lambda} - \theta_{\!\mathcal{A},0} \\[-0.2cm]
 &=  A_{\mathcal{A}}(\theta_{0}) ^{-1}  \frac{1}{n} \sum_{i=1}^n  \left\{ \mu_\mathcal{A}(Y_i, \Delta_i, X_i;\theta_{0})+ \mu_{\mathcal{A},C}^\varphi (Y_i,\Delta_i, X_i;\theta_{0}) \right\}
 + o_{\mathbb{P}}(n^{-1/2}).
 %\\[-0.2cm]
 %& \qquad\qquad\qquad\qquad\qquad+ o_{\mathbb{P}}(n^{-1/2}).%}
\end{align*}%$\\
where
$$
\mu_{\mathcal{A}}(Y,\Delta, X;\theta_{0})  =  \frac{ \left[B^*( S_C)  -\pi (\theta_{0})\right]\omega(X)}{\pi (\theta_{0})[1-\pi (\theta_{0})]} \; \frac{\partial \pi(\theta_{0}) }{\partial \theta_{\mathcal{A}}}
$$
and $\mu_{\mathcal{A},C}^\varphi (Y,\Delta, X;\theta_{0})$ is the zero-mean vector-valued function from Assumption \ref{uclt_ass2} considered with
$$
\varphi_{\mathcal{A}}(X)= \frac{\omega(X)}{ \pi (\theta_{0}) [1 - \pi (\theta_{0})  ]  }  \;  \frac{\partial \pi(\theta_{0}) }{\partial \theta_{\mathcal{A}}}   .
$$
In addition,
$$
\sqrt{n}\left(  \widehat \theta_{\mathcal{A},\lambda} -\theta_{\mathcal{A},0} \right) \rightsquigarrow N_{p_0+1}\left(0, A_{\mathcal{A}}(\theta_{ 0}) ^{-1}V_{\mathcal{A}}(\theta_{ 0}) A_{\mathcal{A}}(\theta_{ 0}) ^{-1}\right)
$$
with $V_{\mathcal{A}}(\theta_{0})= Var\left\{\mu_{\mathcal{A}} (Y,\Delta, X ;\theta_{0}) + \mu_{\mathcal{A},C}^\varphi(Y,\Delta, X ;\theta_{0})\right\}$.
\end{enumerate}
\end{theor}

As was the case for %Assumption  \ref{reg_cure_ass_2},
Theorem \ref{prop_tcl},
we can obtain Theorem \ref{oracle} by imposing condition (\ref{suff_c1}) and slightly stronger regularity conditions on the model $\mathcal{M}$ instead of Assumption {\red\ref{reg_cure_ass_2}.\ref{ULLN_as_nor}}.

\section{Simulation studies\label{sec:sim}}

\subsection{Setup\label{sec:simsetup}}

We first generate $B$, $T_0$ and $C$, from which we obtain $T = T_0 $ when $B=0$ and $T=\infty$ otherwise, and, hence, the observed time, $Y = T \wedge C$, and censoring indicator, $\Delta = (1-B)\ind(T_0 \le C)$, respectively. The cure status is given by $B \sim \text{Bernoulli}(\pi)$ where $\pi (\theta) =  1 / \{1 + \exp(- X^\top  \theta)\}$, $X = (1, X_{(1)}, X_{(2)})^\top$, and $X_{(1)}$ and $X_{(2)}$ are independent $\text{Normal}(0,1)$ variables. We set $\theta_0 = (\theta_{0,(0)}, 1, 1)^\top$ with $\theta_{0,(0)} \in \{-1.85,-0.55\}$ such that the marginal cure proportion $\pi_m = \mathbb{E}\{\pi (\theta)\} \in \{0.2, 0.4\}$.
Consider the survivor function
\begin{equation*}
S_{T_0} (t\mid X) = \left\{\frac{\exp(-t^\kappa) - \exp(-\tau^\kappa)}{1- \exp(-\tau^\kappa)}\right\}^\psi
\end{equation*}
which is that of a truncated Weibull whose support is $(0,\tau)$ with a rate parameter, $\psi$, and a shape parameter, $\kappa$. The latency time, $T_0$, was generated according to this distribution with $\psi = \exp(X^\top \beta_{T_0})$ and $\kappa = (1/\psi)^\nu$ where $\beta_{T_0} = (0, 0, 1)^\top$ and $\nu \in \{0, 2\}$; the proportional hazards property holds when $\nu = 0$. The value of $\tau$ was set at the 95th percentile of the marginal untruncated distribution, i.e., $\tau$ is the unique solution $\mathbb{E}\left\{\exp(-\psi\tau^\kappa) \right\} = 0.05$, and, clearly, $\tau$ depends on the value of $\nu$.
Lastly, the censoring time, $C$, was generated from an exponential distribution with rate parameter $\psi_C = \exp(X^\top \beta_C)$ where $\beta_C = (\beta_{C,(0)}, 0, 1)^\top$. The value $\beta_{C,(0)}$ was chosen such that the overall censored proportion is given by $\pi_\text{cen} = \Pr(\Delta=0) = \pi_m + \rho$ where $\rho \in \{0.1, 0.2\}$, and this depends on the values of $\theta_{(0)}$ and $\nu$; since $\pi_m \in \{0.2, 0.4\}$, there are then four values for the censoring proportion, $\pi_\text{cen} \in \{0.3,0.4,0.5,0.6\}$.

It is worth highlighting that $X_{(1)}$ only affects cure probability (since $\beta_{T_0,(1)}=\beta_{C,(1)}=0$), whereas $X_{(2)}$ affects all components of the data generating process (since $\theta_{(2)}=\beta_{T_0,(2)}=\beta_{C,(2)}=1$). Sample sizes of $n \in \{100, 300, 1000\}$ were considered, and, with two values for each of $\theta_{(0)}$, $\nu$, and $\rho$, there are 24 scenarios altogether. Each simulation scenario was replicated 2000 times. Here we report only on the 6 scenarios where $\pi_m=0.4$ and $\rho=0.1$. The results for the remaining scenarios are broadly similar, {\red as are other scenarios with binary covariates} (see Appendix, Section \ref{sec:simadd}).

\subsection{Estimation procedure}

We applied the estimation scheme described in Section \ref{sec:est} to the simulated data with $S_C$ estimated using a Cox model in which both covariates, $X_{(1)}$ and $X_{(2)}$, appear as predictors. Table \ref{tab:res} displays the average bias and standard error of estimates over simulation replicates. While the bias can be somewhat large when $n=100$, this vanishes as the sample size increases. Similarly, the standard errors also decrease with the sample size. %As we might expect, the estimates generally disimprove when the censoring proportion increases.
Furthermore, the results do not change appreciably when $\nu$ is varied (i.e., the approach is not sensitive to the form of $S_{T_0}$). %, while the standard errors decrease a little when $\pi_m$ is increased. Table \ref{tab:rescov} shows the empirical coverage for 95\% confidence intervals constructed using bootstrapping with 399 replicates; we find that the empirical coverage is close to the nominal level.

\begin{table}[htbp]
\caption{Average bias and standard error (in brackets) of estimates ($\pi_m=0.4,\rho=0.1$)\label{tab:res}}
\centering
%\smallskip
\begin{small}
\begin{tabular}{cc@{~~~~}c@{~~}c@{~~}c@{~~~~}c@{~~}c@{~~}c@{~~~~}c@{~~}c@{~~}c}
\hline
&& \multicolumn{3}{l}{\hspace{0.5cm}$n=100$} &  \multicolumn{3}{l}{\hspace{0.55cm}$n=300$} &  \multicolumn{3}{l}{\hspace{0.55cm}$n=1000$} \\
Method &$\nu$ & $\theta_{(0)}$ & $\theta_{(1)}$ & $\theta_{(2)}$ & $\theta_{(0)}$ & $\theta_{(1)}$ & $\theta_{(2)}$ & $\theta_{(0)}$ & $\theta_{(1)}$ & $\theta_{(2)}$ \\[0.1cm]
\hline
&&&&&&&&&&\\[-0.3cm]
Our
&0 &  -0.05 &   0.10 &   0.12 &  -0.01 &   0.04 &   0.03 &   0.00 &   0.01 &   0.01 \\
proposal
&  & (0.36) & (0.49) & (0.47) & (0.18) & (0.24) & (0.22) & (0.10) & (0.12) & (0.12) \\[0.05cm]
&2 &  -0.05 &   0.13 &   0.11 &  -0.02 &   0.04 &   0.02 &   0.00 &   0.01 &   0.01 \\
&  & (0.37) & (0.53) & (0.44) & (0.19) & (0.25) & (0.21) & (0.10) & (0.13) & (0.11) \\[0.1cm]
\texttt{smcure}
&0 &  -0.08 &   0.09 &   0.09 &  -0.02 &   0.03 &   0.02 &  -0.01 &   0.01 &   0.01 \\
&  & (0.33) & (0.40) & (0.38) & (0.17) & (0.20) & (0.19) & (0.09) & (0.11) & (0.10) \\[0.05cm]
&2 &  -0.12 &   0.10 &  -0.09 &  -0.08 &   0.04 &  -0.13 &  -0.06 &   0.01 &  -0.13 \\
&  & (0.35) & (0.40) & (0.33) & (0.19) & (0.21) & (0.17) & (0.10) & (0.11) & (0.09) \\[0.1cm]
\hline
% &&\\[-0.2cm]
\end{tabular}
\end{small}
\end{table}

By way of comparison, we also applied the EM approach of \citet{pengdear:2000} and \citet{syetal:2000} which has been implemented in the \texttt{smcure} \citep{chaoetal:2012} package in \texttt{R} \citep{R:2018}. In contrast to our scheme, $S_{T_0}$, rather than $S_C$, must be estimated. Thus, $S_{T_0}$ was estimated using a Cox model in which both covariates, $X_{(1)}$ and $X_{(2)}$, appear as predictors. The results, also shown in Table \ref{tab:res}, are similar to those of our proposal when $\nu=0$. However, when $\nu = 2$ (i.e., $S_{T_0}$ does not have the proportional hazards property), we see bias in the \texttt{smcure} estimates which does not disappear with increasing sample size. In particular, the bias manifests through $\widehat\theta_{(0)}$ and $\widehat\theta_{(2)}$; interestingly, $\widehat\theta_{(1)}$ is unaffected (i.e., the coefficient of $X_{(1)}$, the covariate which only enters the cure component).

{\red As discussed in Section \ref{sec:normal}, the setting we consider in this paper lies within the theory of \citet{hitomi_nishiyama_okui_2008}, suggesting that the variability of the estimates will be higher when the true $S_C$ is used for estimation as per \eqref{likedata} than when $\widehat S_C$ is used as per \eqref{likedata_b}, and, in turn, it will be lower still if the true cure labels are used as per \eqref{likeB}. Although only the $\widehat S_C$ case is feasible, we display the standard errors for all three cases in Table \ref{tab:Bs}, and the results for the other simulation scenarios are given in the Appendix (Section \ref{sec:simadd}). Separately, but not shown, we have estimated $S_C$ parametrically using maximum likelihood, and considered additional scenarios where the true $S_C$ does not depend on $X$ (estimating $S_C$ both parametrically and non-parametrically); in all cases we have considered, the results are as predicted by \citet{hitomi_nishiyama_okui_2008}.}

\begin{table}[htbp]
\caption{Standard error of estimates using different $B$ types ($\pi_m=0.4,\rho=0.1$)\label{tab:Bs}}
\centering
\begin{small}
\begin{tabular}{cc@{~~~~}c@{~~}c@{~~}c@{~~~~}c@{~~}c@{~~}c@{~~~~}c@{~~}c@{~~}c}
\hline
&& \multicolumn{3}{l}{\hspace{0.4cm}$n=100$} &  \multicolumn{3}{l}{\hspace{0.4cm}$n=300$} &  \multicolumn{3}{l}{\hspace{0.4cm}$n=1000$} \\
$\nu$ & $B$ Type & $\theta_{(0)}$ & $\theta_{(1)}$ & $\theta_{(2)}$ & $\theta_{(0)}$ & $\theta_{(1)}$ & $\theta_{(2)}$ & $\theta_{(0)}$ & $\theta_{(1)}$ & $\theta_{(2)}$ \\[0.1cm]
\hline
&&&&&&&&&&\\[-0.3cm]
0  &   $B^*(S_C)$      & (0.40) & (0.53) & (0.50) & (0.22) & (0.28) & (0.26) & (0.12) & (0.14) & (0.13) \\
   &   $B^*(\widehat S_C)$ & (0.36) & (0.49) & (0.47) & (0.18) & (0.24) & (0.22) & (0.10) & (0.12) & (0.12) \\[0.03cm]
   &   $B$             & (0.25) & (0.32) & (0.30) & (0.14) & (0.17) & (0.17) & (0.08) & (0.09) & (0.09) \\[0.13cm]
2  &   $B^*(S_C)$      & (0.45) & (0.58) & (0.50) & (0.23) & (0.29) & (0.25) & (0.12) & (0.14) & (0.12) \\
   &   $B^*(\widehat S_C)$ & (0.37) & (0.53) & (0.44) & (0.19) & (0.25) & (0.21) & (0.10) & (0.13) & (0.11) \\[0.03cm]
   &   $B$             & (0.26) & (0.30) & (0.30) & (0.14) & (0.17) & (0.17) & (0.08) & (0.09) & (0.09) \\[0.08cm]
\hline
% &&\\[-0.2cm]
\end{tabular}
\end{small}
\end{table}

\subsection{Selection procedure\label{sec:simselect}}

We simulated data as per Section \ref{sec:simsetup}, but with four additional independent $\text{Normal}(0,1)$ variables, $X_{(3)}$, $X_{(4)}$,  $X_{(5)}$, and $X_{(6)}$. These variables do not affect the cure probability, i.e., their $\theta$ coefficients are zero, but we set $\beta_{T_0,(3)}=\beta_{C,(3)}=1$ so that $X_{(3)}$ affects other aspects of the data generating process; the $\beta_{T_0}$ and $\beta_{C}$ coefficients for $X_{(4)}$,  $X_{(5)}$, and $X_{(6)}$ are all zero. {\red The bias and standard errors for this setup can be found in the Appendix (Section \ref{sec:simadd}), and, expectedly, they are a little larger than the scenarios with two covariates. However, here, we focus on the results of variable selection} using the alasso where $\lambda$ was chosen by minimizing the cross-validation error (see Appendix, Section \ref{sec:imp} for {\red details of this algorithm}). Let $\lambda^{\text{CVE}}$ denote this $\lambda$ value, and define the following commonly-used metrics for assessing the selection performance: $\text{C}  = \sum_{j = 4}^6 \ind(\widehat\theta_{\lambda^{\text{CVE}},(j)} = 0)$, the number of coefficients \emph{correctly} set to zero, $\text{IC} = \sum_{j = 1}^2 \ind(\widehat\theta_{\lambda^{\text{CVE}},(j)} = 0)$, the number of coefficients \emph{incorrectly} set to zero, and $\text{DF} = \sum_{j = 0}^6 \ind(\widehat\theta_{\lambda^{\text{CVE}},(j)} > 0)$,
%\end{align*}
%\begin{align*}
%\text{C}  &= \sum_{j = 4}^6 \ind(\widehat\theta_{\widehat\lambda_{\text{CVE}},(j)} = 0), \qquad
%&\text{IC} &= \sum_{j = 1}^2 \ind(\widehat\theta_{\widehat\lambda_{\text{CVE}},(j)} = 0), \qquad
%&\text{DF} &= \sum_{j = 0}^6 \ind(\widehat\theta_{\widehat\lambda_{\text{CVE}},(j)} > 0),
%\end{align*}
the model \emph{degrees of freedom} (i.e., the number of non-zero parameters); in our setup, for the oracle model, C $= 4$, IC $= 0$, and DF $=3$. These metrics, averaged over simulation replicates, are shown for the alasso in Table \ref{tab:resvar}. In line with what we expect from Theorem \ref{oracle}, IC approaches zero as the sample size increases, while C approaches four; again the results are unaffected by $\nu$. %, whereas, when $n=100$, increased censoring proportion, $\rho$, or decreased cure proportion, $\pi_m$, both lead to fewer variables being selected.
%Note that the lasso (results in ) tends to select an overly complicated model which is expected behaviour \citep{fanlv:2010}.

\begin{table}[htbp]
\caption{Correct zeros, incorrect zeros, and model degrees of freedom ($\pi_m=0.4,\rho=0.1$)\label{tab:resvar}}
\centering
\begin{small}
\smallskip
\begin{tabular}{cc@{~~~~}c@{~~}c@{~~}c@{~~~~}c@{~~}c@{~~}c@{~~~~}c@{~~}c@{~~}c}
\hline
&& \multicolumn{3}{l}{\hspace{0.3cm}$n=100$} &  \multicolumn{3}{l}{\hspace{0.3cm}$n=300$} &  \multicolumn{3}{l}{\hspace{0.3cm}$n=1000$} \\
Type & $\nu$ & C & IC & DF & C & IC & DF & C & IC & DF \\[0.1cm]
\hline
&&&&&&&&&&\\[-0.3cm]
oracle &  & 4.00 & 0.00 & 3.00 & 4.00 & 0.00 & 3.00 & 4.00 & 0.00 & 3.00\\[0.1cm]
%lasso   & 0 & 2.52 & 0.11 & 4.37 & 2.33 & 0.00 & 4.67 & 2.30 & 0.00 & 4.70 \\[0.0cm]
%        & 2 & 2.58 & 0.10 & 4.32 & 2.37 & 0.00 & 4.63 & 2.24 & 0.00 & 4.76 \\[0.1cm]
alasso & 0 & 3.34 & 0.17 & 3.49 & 3.52 & 0.00 & 3.48 & 3.69 & 0.00 & 3.31 \\[0.0cm]
       & 2 & 3.33 & 0.15 & 3.52 & 3.55 & 0.00 & 3.44 & 3.72 & 0.00 & 3.28 \\[0.1cm]
\hline
% &&\\[-0.2cm]
\end{tabular}
\end{small}
\end{table}

\section{Data analysis: colon cancer\label{sec:data}}

%\subsection{Overview of data}

We consider a colon cancer dataset (contained in the \texttt{survival} package in \texttt{R}) which was collected as part of a well-known national intergroup randomized controlled trial (involving Eastern Cooperative Oncology Group, the North Central Cancer Treatment Group, the Southwest Oncology Group, and the Mayo Clinic). The aim of the study was to investigate the efficacy of the drugs levamisole and 5FU for the treatment of colon cancer following surgery; relapse-free survival was the outcome variable of interest, i.e., time from randomization until the earlier of cancer relapse or death. In total, 929 patients with stage C disease enrolled during the period March 1984 to October 1987, with a maximum follow-up time of nine years. These patients were randomized to the following treatments: observation (control / reference group), levamisole, and a combined treatment of levamisole and 5FU.
In addition to the treatment variable, a variety of binary covariates were recorded (reference categories are shown first): days since surgery, $\{\le20,>20\}$; sex, $\{\text{female},\text{male}\}$; obstruction of colon by tumour, $\{\text{no},\text{yes}\}$; adherence to nearby organs, $\{\text{no},\text{yes}\}$; depth of invasion, $\{\text{submucosa or muscular layer} ,\text{serosa}\}$; positive lymph nodes, $\{\le4,>4\}$. Furthermore, the age of the patient was recorded, and we use a mean-centered version (mean age is 59.75 years). See \citet{moertel:1990} for further details. This dataset is a candidate for cure analysis based on its Kaplan-Meier (KM) curve which has a clear plateau at approximately 40\% (Figure \ref{fig:kmcurves}). %Note that the last value in the KM curve is an estimator of the marginal cure probability \citep{mallerzhou:1992} and our approach reduces to this when there are no covariates (i.e., $X \equiv 1$); see \citet{satten:2001} for details.

\begin{figure}%[htbp]
\begin{center}
\begin{tabular}{c}
\includegraphics[width=0.73\textwidth, trim = {0cm 0.7cm 0.5cm 2cm}, clip]{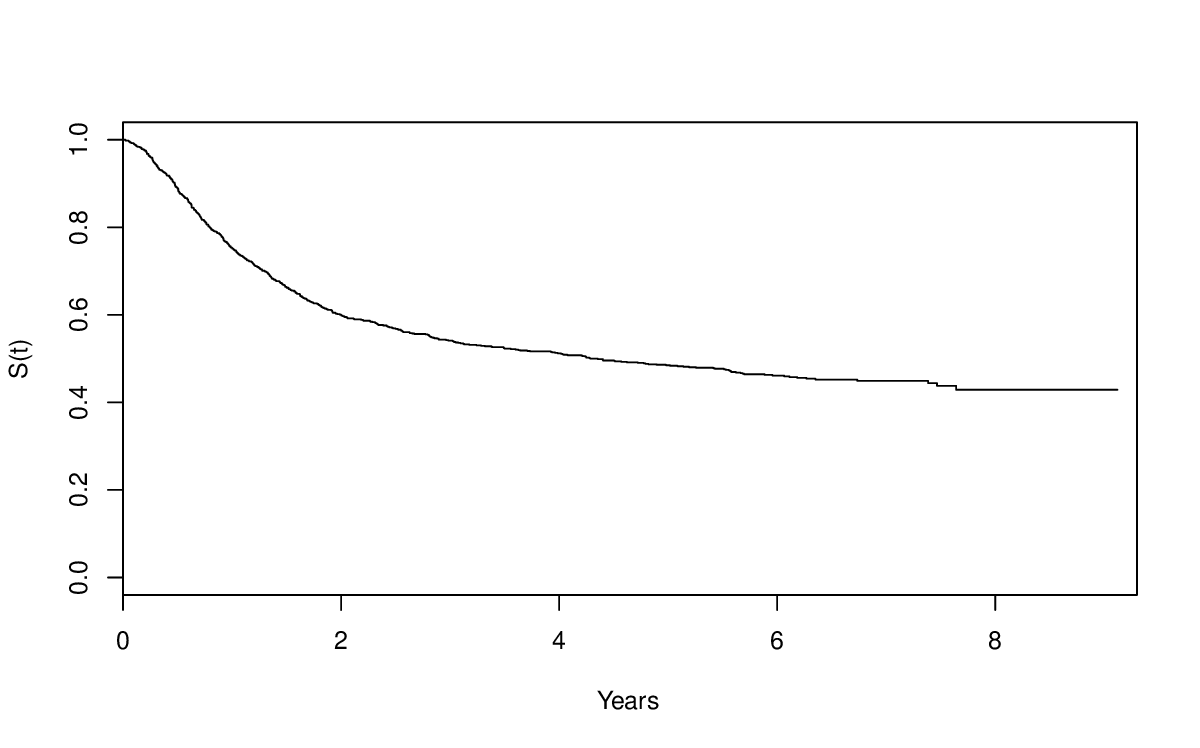}
\end{tabular}
\caption{Kaplan-Meier curve for colon data.\label{fig:kmcurves}}
\end{center}
\vspace{-0.5cm}
\end{figure}

We estimate the cure parameters using our proposed procedure and, for comparison, apply \texttt{smcure}. % with signs of the coefficients reversed (to align with our model for the cure probability, rather than the non-cure probability).
We use a Cox model with all covariates for $\widehat S_C$ in our approach and for $\widehat S_{T_0}$ in \texttt{smcure}, and a logistic regression cure model in both cases; confidence intervals and p-values are produced using bootstrapping. We also carry out variable selection using the adaptive lasso where covariates are standardized for variable selection, after which the estimates are transformed back to correspond to the original scale. See Table \ref{tab:colon}.

%\subsection{Colon cancer data\label{sec:colon}}

\begin{table}[htbp]
\caption{Colon cancer estimates\label{tab:colon}}
\centering
\begin{footnotesize}
\smallskip
\begin{small}
\begin{tabular}{l@{~~}l@{~\quad}c@{}r@{~~}c@{~~}c@{~~}c@{}r@{~~}r@{~~}c@{~~\quad}r@{~~}c@{~~}c@{}}
%\begin{tabular}{cc@{~\qquad}c@{}c@{~~~}c@{~~~}c@{~\qquad}c@{}c@{~~~}c@{~~~}c@{~~~}c@{~~~}c@{~~~}c@{~~~}c@{}}
\hline
&&&&&&&&&&&&\\[-0.3cm]
 & & &  \multicolumn{3}{c}{Unpenalized} && \multicolumn{2}{c}{alasso} && \multicolumn{3}{c}{\texttt{smcure}}\\[0.1cm]
\multicolumn{2}{c}{Covariate} & &  Est. &  95\%CI &  pval    &   &   &  Est. &    &  Est. & 95\%CI &  pval   \\
\hline
&&&&&&&&&&&&\\[-0.3cm]
Intercept      &           &  &  0.66 & ( 0.05, 1.36) & 0.03  &    & &  0.32 &  &  0.57 & (-0.05, 1.15) & 0.06  \\
Treatment      & Lev       &  &  0.42 & (-0.11, 1.22) & 0.14  &    & &  0.00 &  &  0.19 & (-0.21, 0.61) & 0.33  \\
               & Lev+5FU   &  &  0.94 & ( 0.39, 1.73) & 0.00  &    & &  0.60 &  &  0.71 & ( 0.30, 1.15) & 0.00  \\
Surgery        & $>20$days &  & -0.65 & (-1.63,-0.11) & 0.02  &    & & -0.41 &  & -0.49 & (-0.88,-0.09) & 0.01  \\
Age           & Years     &  & -0.01 & (-0.03, 0.00) & 0.10  &    & &  0.00 &  & -0.01 & (-0.02, 0.00) & 0.24  \\
Sex            & Male      &  & -0.24 & (-0.75, 0.16) & 0.28  &    & &  0.00 &  & -0.11 & (-0.46, 0.25) & 0.60  \\
Obstruction    & Yes       &  & -0.56 & (-2.01, 0.05) & 0.08  &    & & -0.19 &  & -0.18 & (-0.61, 0.20) & 0.35  \\
Adherence      & Yes       &  & -0.42 & (-1.00, 0.07) & 0.10  &    & &  0.00 &  & -0.69 & (-1.50,-0.17) & 0.01  \\
Depth          & Serosa    &  & -0.81 & (-1.51,-0.26) & 0.01  &    & & -0.54 &  & -0.71 & (-1.27,-0.19) & 0.02  \\
Nodes          & $>4$      &  & -1.18 & (-1.63,-0.82) & 0.00  &    & & -0.94 &  & -1.16 & (-1.59,-0.80) & 0.00  \\[0.1cm]
\hline
&&&&&&&&&&&&\\[-0.3cm]
\multicolumn{13}{p{0.92\textwidth}}{\footnotesize Age is mean-centered.}
% &&\\[-0.2cm]
% us: 33.29 Z
% us:  99.69 Xorig
% smcure: 1118.28 Z
%% smcure: 1786.36 Xorig
\end{tabular}
\end{small}
\end{footnotesize}
\end{table}
\begin{figure}[htbp]
\begin{center}
\begin{tabular}{c}
\includegraphics[width=0.9\textwidth, trim = {0.2cm 0.6cm 0.7cm 1cm}, clip]{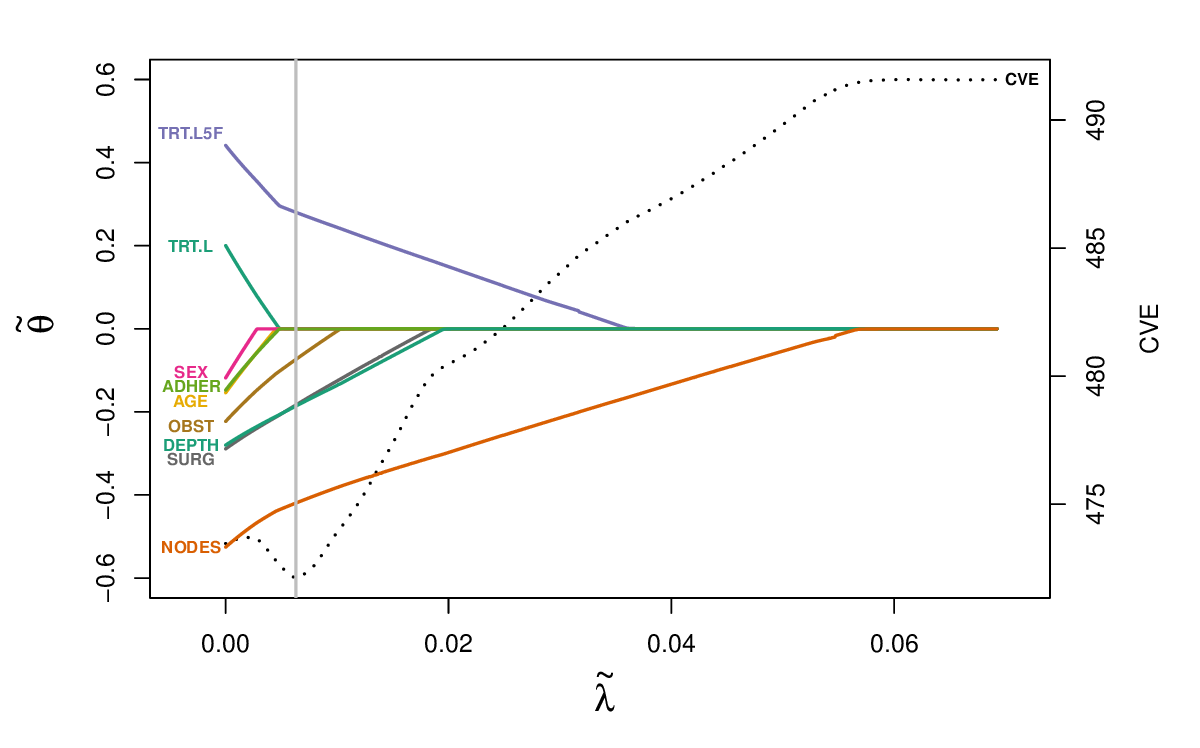}
\end{tabular}
\caption{Adaptive lasso regularization paths for colon data. Estimates, denoted by $\widetilde \theta$, are for the standardized covariates (hence, their magnitudes can be compared), and are plotted against the tuning parameter $\widetilde\lambda = \lambda/n$. Also shown is the 10-fold cross-validation error curve (dotted, and see right-hand $y$-axis) with vertical line indicating its minimum.\label{fig:regcolon}}
\end{center}
\end{figure}

First we consider the the unpenalized estimates. The effect of the levamisole treatment does not significantly increase the cure probability (compared with a patient receiving no treatment), while the combination of levamisole with 5FU does; indeed, the odds of being cured for this latter treatment are $2.56$ ($=\exp(0.94)$) with 95\% confidence interval given by $(1.48,5.64)$. The effect of all other covariates is to reduce the cure probability, albeit sex is not statistically significant, and obstruction and adherence are only just significant at the 10\% level. The results for \texttt{smcure} are broadly similar, apart from the fact that adherence is statistically significant. Now, turning to the penalized alasso estimates, several coefficients have been set to zero, and the retained variables are those with smaller p-values from the unpenalized model. The regularization paths for standardized cure coefficients (i.e., those corresponding to standardized covariates) provide useful information on the relative importance of each covariate; these are shown in Figure \ref{fig:regcolon}. We can see immediately that the Lev+5FU treatment is one of the most important features. The number of positive lymph nodes is also highly important, and the presence of more than four such lymph nodes reduces the chance of cure. Next, the timing of surgery and depth of the tumour have similar importance, followed by the presence of an obstruction.

\section{Discussion\label{sec:disc}}

We have proposed an {\red IPCW-likelihood-based} estimation procedure for cure regression models; elsewhere IPCW has been advocated by \citet{KMintegral_laan} as a device for producing straightforward estimators in complex survival data. In contrast to current cure estimation procedures in the literature, our assumptions are placed on $S_C$ while $S_{T_0}$ is completely unspecified. Although we have considered a Cox model estimator for $S_C$ in the examples in this article, any arbitrarily flexible model can be used in practice as this simply ``plugs in'' to the likelihood function given in (\ref{likedata_b}) without any added complexity to the estimation procedure. Moreover, our asymptotic results still hold once the estimator, $\widehat S_C$, permits an iid representation (and we have given many common examples in the Appendix, Section \ref{sec:unifiid}).

Except for the case of a fully nonparametric approach like in \citet{xu:2014} (which suffers from the curse of dimensionality), existing cure regression models impose assumptions on both the cure proportion and the law of the susceptible individuals, without satisfactory model diagnosis (besides ad-hoc efforts). In our approach, one can first use standard diagnosis procedures to validate the censoring model as this is identifiable from the observed data directly. For example, one could assess the proportional hazards assumption for $S_C$ using the test due to  \citet{grambschthern:1994} which is implemented in the \texttt{cox.zph} function in the \texttt{survival} package in \texttt{R}. (Although not shown, this test supported the proportional hazards assumption in the application considered in Section \ref{sec:data}.) Next one could consider model diagnostics for the cure regression. %, e.g., based on error terms of the form $B_i^*(\widehat S_C) - \pi_i(\widehat\theta)$ which we made use of in (\ref{cvmse}).
Furthermore, note that our theory is not limited to the logistic model choice used in our applications, and, more generally still, the functional form of the cure regression model $\pi(X^\top\theta)$ could itself be estimated (e.g., in a similar manner to \citet{amicoetal:2018} who extended the EM approach in this way). Goodness-of-fit for the cure model and estimation of its functional form are beyond the scope of the current article.

 % but will be developed in our future work.

Although the extension to penalized estimation is straightforward and computationally efficient in our setting, we note that penalized selection in the existing EM setting was considered by \citet{liuetal:2012}. However, those authors remark on the computational intensiveness of the procedure (penalized estimation and tuning parameter selection are embedded in EM). Furthermore, their approach is limited to a Cox PH model for $S_{T_0}$, and lacks asymptotic theory.

{\red Lastly, our focus has been on modelling the cure probability without specifying a model for the latency distribution. However, as pointed out by an anonymous reviewer, the latter might also be of interest in applications. Indeed, from the fact that $S_T(t\mid X) = \{1-\pi(X)\} S_{T_0}(t\mid X) + \pi(X)$, we have
$S_{T_0}(t\mid X) = \{S_T(t\mid X) - \pi(X)\}/\{1-\pi(X)\}$
which immediately provides $\widehat S_{T_0}$ given $\widehat \pi$ and $\widehat S_T$. Of course, $S_T$ can be modelled in an arbitrarily flexible way as it is based directly on the observable data, and any of the standard survival models can simply ``plug in'' just as for the $S_C$ model in our framework.}

\newpage

\bibliographystyle{apalike}
\bibliography{refs}

\newpage

\section{Appendix\label{sec:app}}
\setcounter{equation}{0}

\subsection{Proofs\label{sec:proofs}}

\begin{lemma}\label{core_id}
Let $r(Y,X)$ be an integrable real-valued function. Under conditions (\ref{ass1}) and (\ref{ass2}),
$$
\expect\left(\left.\frac{\Delta r(Y,X)}{S_C(Y-\mid X)} ~\right| X\right)~=~ \expect\{r(T_0,X)\mid X\} \{1-\pi(X)\}.
$$
\end{lemma}

\begin{proof}[Proof of Lemma  \ref{core_id}]
First, we have that
\begin{align*}
\expect(\Delta \mid T_0, X ) &= \expect\{\ind(T_0 \le C) (1-B) \mid T_0, X \} \notag\\
&= \expect\{\ind(T_0 \le C)\expect(1-B \mid C, T_0, X ) \mid T_0, X \} \notag\\
%&= \expect\{\ind(T_0 \le C) \mid T_0, X \}\,\expect(1-B \mid X) \notag\\
&= S_C(T_0-\mid X) \,\{1-\pi(X)\} ,%\label{expectiter}
\end{align*}
where $\expect(1-B \mid C, T_0, X ) = \expect(1-B \mid X) = 1-\pi(X)$ follows from (\ref{ass2}), and $\expect[\ind(T_0 \le C) \mid T_0, X ] = S_C(T_0-\mid X)$ follows from (\ref{ass1}). Thus,
\begin{align*}
\expect\left(\left.\frac{\Delta r(Y,X)}{S_C(Y-\mid X)} ~\right| X\right)
&= \expect\left(\left.\frac{\Delta r(T_0,X)}{S_C(T_0-\mid X)} ~\right| X\right) \\
&= \expect\left\{\left. \frac{ r(T_0,X)}{S_C(T_0-\mid X)} \expect\left(\left.\Delta ~\right| T_0, X\right)  ~\right| X\right\} \\
%&= \expect\left\{\left.  r(T_0,X) \expect(1-B \mid X)  ~\right| X\right\} \\
&=\expect\{r(T_0,X)\mid X\} \{1-\pi(X)\},
\end{align*}
as required. (The first equality holds since only the $\Delta=1$ case contributes where $Y=T_0$.)
\end{proof}

\begin{lemma}\label{indep_eq}
Let $B$ be a Bernoulli random variable, $T_0$ a nonnegative random variable  and let
$T =  T_0$ if $B=0$ and $T=\infty$ if $B=1$. Then
$$
T_0 \perp C   \mid X  \text{ and } B \perp (T_0, C)   \mid X  \quad \Longrightarrow \quad T \perp C   \mid X .
$$
\end{lemma}

\begin{proof}[Proof of Lemma \ref{indep_eq}]
By elementary properties of the conditional independence
$$
B \perp (T_0, C)   \mid X \quad \Longleftrightarrow \quad B \perp C  \mid (X, T_0) \quad \text{ and } \quad  B \perp T_0 \mid X.
$$
Next,
$$
B \perp C  \mid (X, T_0) \quad \text{ and } \quad T_0 \perp C   \mid X \quad \Longleftrightarrow \quad  (B,T_0) \perp C   \mid X .
$$
The result follows from the fact that $T$ is completely determined by $B$ and $T_0$. \end{proof}

\quad

\begin{lemma}\label{lik_ineq_proof}
Let
\begin{equation*}
\ell^*(\theta) = \sum_{i=1}^n\left[B_i^*(S_C) \log \pi_i + \{1- B_i^*(S_C)\} \log(1-\pi_i)\right]   \omega_i .
\end{equation*}
If the cure regression model is identifiable,
\begin{equation*}%\label{entropy_ineg}
\expect\{\ell^*(\theta)\} < \expect\{\ell^*(\theta_0)\} \qquad \forall \theta \ne \theta_0
\end{equation*}
\end{lemma}

\begin{proof}[Proof of Lemma  \ref{lik_ineq_proof}]
Let $\pi^* (X) = \mathbb{E}\{B_i^*(S_C)  \mid X\}$, which by construction lies between 0 and 1. Then, since for any $u>0$,  $\log (u) \leq u-1$, we deduce
\begin{multline*}
\expect\{\ell^*(\theta)\} - \expect\{\ell^*(\theta_0)\} = \expect\{\expect\{\ell^*(\theta)\}\mid X\}  - \expect\{\expect\{\ell^*(\theta_0)\}\mid X \} \\ = \expect\left\{\left[\pi^* (X) \log \frac{\pi (X)}{\pi^* (X)} + \{  1 - \pi^* (X) \} \log \frac{1- \pi (X)}{1- \pi^* (X)}\right]\omega(X)\right\}\\
\leq \expect\left[\pi^* (X) \left\{ \frac{\pi (X)}{\pi^* (X)} -1 \right\}\omega(X)  + \{  1 - \pi^* (X) \} \left\{ \frac{1- \pi (X)}{1- \pi^* (X)}-1 \right\} \omega(X) \right]=0,
\end{multline*}
and the  inequality is strict unless $[\pi (X)-\pi^* (X ) ]\omega(X)=0 $ almost surely. If the cure model is identifiable, this cannot happen and thus
$\expect\{\ell^*(\theta)\} < \expect\{\ell^*(\theta_0)\},$ $\forall \theta \ne \theta_0.$
\end{proof}

\quad

\begin{proof}[Proof of Lemma \ref{suff_cdt_ulln}]
Define the event
\begin{multline*}
\mathcal E_n = \left\{  \sup_{x \in\mathcal{X}, \; \omega(x) >0} \;\sup_{y\leq\tau(x)}\left|\widehat S_C(y-\mid x) - S_C(y-\mid x)  \right|\right. \\
\left. \leq (1/2) \inf_{x \in\mathcal{X}, \; \omega(x) >0} S_C(\tau(x)-\mid x)  \right\}.
\end{multline*}
On the set $\mathcal E_n$, given a  measurable function $\phi(\cdot),$ we can write
\begin{multline*}
\left|  \frac{1}{n} \sum_{1\leq i\leq n} \left[ B^*(\widehat S_C) - B^*( S_C) \right]\omega(X_i) \phi (X_i) \right|  \\\leq
 \frac{1}{n} \sum_{1\leq i\leq n}\frac{  \Delta_i \;  | \widehat R_i |\;\omega(X_i) \; |\phi (X_i) |}  {S_C(Y_i-\mid  X_i)\left[ S_C(Y_i-\mid  X_i)- | \widehat R_i | \right] } \\ \leq
2  \sup_{1\leq i\leq n}   | \widehat R_i |  \times \frac{1}{n} \sum_{1\leq i\leq n}\frac{ \omega(X_i)  \; |\phi (X_i) |}  {S^2_C(Y_i-\mid  X_i) },
\end{multline*}
where
$
\widehat R_i = \widehat S_C(Y_i-\mid  X_i) -  S_C(Y_i-\mid  X_i).
$
By the preservation of the Glivenko-Cantelli property for classes of functions (see Theorem 3 of \citet{vaart_wellner2000}), the set of logit transformations of $\pi(\theta)$ is a  $\mathbb{P}_{X}-$Glivenko-Cantelli class of functions of $X$ with constant envelope, provided
 Assumption {\red\ref{reg_cure_ass}.\ref{ass_truth}} and {\red\ref{reg_cure_ass}.\ref{ass_GC}} hold true.
The  statement follows from the uniform law of large numbers for the empirical process indexed by the set of  functions\\
$\left\{\omega(\cdot)S^{-2}_C(\cdot \; -\mid \cdot)|\log\{\pi(\cdot;\theta)/[1-\pi(\cdot;\theta)]\}:\theta\in\Theta\right\} $ with constant envelope, and condition (\ref{suff_c1}) which implies $\sup_{1\leq i\leq n}   | \widehat R_i | = o_{\mathbb{P}}(1) $ and $\Pr(\mathcal E_n)\rightarrow 1$.
\end{proof}

\quad

\begin{proof}[Proof of Theorem \ref{conv_th1}]
We apply Theorem 5.7 of \citet{van2000asymptotic}. First, by construction and Assumption {\red\ref{reg_cure_ass}.\ref{ass_sep}},  $\theta_0$ is a well-separated maximum of  the map $\theta \mapsto \expect\{\ell^* (\theta)\},$ with $\ell^*(\theta)$ defined in equation (\ref{likedata}). See the inequality of Lemma \ref{lik_ineq_proof}. %See inequality (\ref{entropy_ineg}).
Next, let us note that
$$
\widehat \ell^*(\theta) - \ell^*(\theta) =  \sum_{i=1}^n \left[ B^*(\widehat S_C) - B^*( S_C) \right] \omega(X_i) \log\left( \frac{\pi_i(\theta)}{1-\pi_i(\theta)} \right) .
$$
Since
$$
\sup_{\theta\in\Theta }\left| \widehat \ell^*(\theta)-  \expect\{\ell^* (\theta)\} \right| \leq
\sup_{\theta\in\Theta } \left| \widehat \ell^*(\theta)-   \ell ^* (\theta) \right| + \sup_{\theta\in\Theta } \left| \ell ^* (\theta)-  \expect\{\ell^* (\theta)\} \right|
$$
by our assumptions and the uniform law of large numbers for the empirical process,
$$
\sup_{\theta\in\Theta }\left| \widehat \ell^*(\theta)-  \expect\{\ell^* (\theta)\} \right| = o_{\mathbb{P}}(1).
$$
The consistency of $\widehat \theta $ follows from Theorem 5.7 of \citet{van2000asymptotic}.
\end{proof}

\quad

\begin{proof}[Proof of Theorem \ref{prop_tcl}]
By the definition of $\widehat \theta$ we have the identity
$\widehat \ell^*(\theta)$
\begin{multline*}
0 = \frac{\partial \widehat \ell^* (\widehat \theta)}{\partial \theta}   =   \frac{\partial \widehat \ell^*  (\theta_0)}{\partial \theta}  +  \frac{\partial^2 \widehat \ell^*  (\overline \theta) }{\partial \theta \partial \theta ^\top }\left( \widehat \theta - \theta_0 \right) \\ =
  \frac{\partial \widehat \ell^*  (\theta_0)}{\partial \theta}  + \frac{\partial^2 \widehat \ell^*  (\theta_0) }{\partial \theta \partial \theta ^\top } \left( \widehat \theta - \theta_0 \right)   + \left[  \frac{\partial^2 \widehat \ell^*  (\overline \theta) }{\partial \theta \partial \theta ^\top } -   \frac{\partial^2 \widehat \ell^*  ( \theta_0) }{\partial \theta \partial \theta ^\top } \right]  \left( \widehat \theta - \theta_0 \right),
\end{multline*}
where $\overline \theta$ is some point on the segment between $\widehat \theta$ and $\theta_0$.
Using the definition in equation (\ref{def_Bistar}) and  the short notation $\pi_i(\theta ) = \pi (X_i;\theta)$ and $\partial \pi_i(\theta)/ \partial \theta= \partial \pi(X_i; \theta)/ \partial \theta$, we can write
\begin{multline*}
\frac{1}{n} \frac{\partial \widehat \ell^*  (\theta_0)}{\partial \theta} = \frac{1}{n}\sum_{i=1} ^n \left[ \frac{B_i^*(\widehat S_C)}{\pi _i (\theta_0) } -  \frac{1-B_i^*(\widehat S_C)}{1- \pi _i (\theta_0) }\right] \omega_i  \frac{\partial \pi_i (\theta_0) }{\partial \theta}    \\
= \frac{1}{n} \sum_{i=1} ^n  \frac{B_i^*( S_C) - \pi _i (\theta_0)}{\pi_i (\theta_0) [1- \pi _i (\theta_0)]}\;  \omega_i   \frac{\partial \pi_i (\theta_0)}{\partial \theta}   \\
+ \frac{1}{n} \sum_{i=1} ^n  \frac{B_i^*(\widehat S_C) - B_i^*(S_C)}{\pi _i (\theta_0) [1- \pi _i (\theta_0)]}\; \omega_i \frac{\partial \pi_i (\theta_0)}{\partial \theta}   ,
\end{multline*}
and thus, by Assumption \ref{uclt_ass2} we have
$$
\frac{1}{n} \frac{\partial \widehat \ell^*  (\theta_0)}{\partial \theta} =  \frac{1}{n} \sum_{i=1}^n  \left\{ \mu(Y_i,\Delta_i, X_i;\theta_0) + \mu_C^\varphi (Y_i,\Delta_i, X_i;\theta_0) \right\}   + o_{\mathbb{P}}(n^{-1/2}).
$$
Next, for any $\theta$ we have
\begin{multline*}
\frac{1}{n}\frac{\partial^2 \widehat \ell^*  ( \theta) }{\partial \theta \partial \theta ^\top } = - \frac{1}{n} \sum_{i=1} ^n  \left[ \frac{B_i^*(\widehat S_C)}{\pi _i ^2 (\theta) } +  \frac{1-B_i^*(\widehat S_C)}{[1- \pi _i (\theta)]^2 } \right] \omega_i  \frac{\partial \pi_i (\theta)}{\partial \theta} \frac{\partial \pi_i (\theta)}{\partial \theta}^\top   \\+  \frac{1}{n} \sum_{i=1} ^n  \frac{B_i^*( \widehat S_C) - \pi _i (\theta)}{\pi _i(\theta) [1- \pi _i (\theta)]} \; \omega_i  \frac{\partial ^2\pi_i (\theta)}{\partial \theta\partial \theta^\top }  \\
\stackrel{def}{=} - H_{1n}(\widehat S_C; \theta )+H_{2n}(\widehat S_C; \theta).
\end{multline*}
Here, $\partial ^2\pi_i (\theta)/ \partial \theta\partial \theta^\top$ denotes the matrix of the second order partial derivatives of  $\pi(X_i;\theta)$ with respect to $\theta$. By the preservation of the Glivenko-Cantelli property for classes of functions (see Theorem 3 of \citet{vaart_wellner2000}), and {\red Assumptions \ref{reg_cure_ass}.\ref{ass_truth} and  \ref{reg_cure_ass}.\ref{ass_GC}}, the sets
$$
\left\{ \pi(\theta)^{-2} \omega  f , \;\;[1-\pi(\theta)]^{-2} \omega  f:f\in\mathcal{F}_{1,kl}, \theta \in \Theta\right\}
$$
and
$$
\left\{ \{ \pi(\theta)  [1-\pi(\theta)]  \}^{-1} \omega  g , \;\;  [1-\pi(\theta) ]  ^{-1} \omega  g  :g\in\mathcal{F}_{2,kl}, \theta \in \Theta\right\}
$$
with $\mathcal{F}_{1,kl}$ and $\mathcal{F}_{2,kl},$ $0\leq k,l\leq p,$ defined in Assumption {\red\ref{reg_cure_ass_2}.\ref{ULLN_as_nor}},
are $\mathbb{P}_{X}-$Glivenko-Cantelli classes of functions of $X$ with integrable envelope. By Lemma \ref{suff_cdt_ulln}, deduce that
$$
\sup_{\theta\in\Theta }\left\{ \left| H_{1n}(\widehat S_C; \theta )-H_{1n}( S_C; \theta )\right| +  \left| H_{2n}(\widehat S_C; \theta )-H_{2n}( S_C; \theta )\right| \right\} = o_{\mathbb P}(1).
$$
On the other hand, the uniform law of large numbers for the empirical process yields
$$
\sup_{\theta\in\Theta }\left\{ \; \left| H_{1n}(S_C; \theta )- \mathbb{E}\{ H_{1n}( S_C; \theta )\} \right| +  \left| H_{2n}(S_C; \theta )- \mathbb{E}\{ H_{2n}( S_C; \theta )\} \right| \; \right\} = o_{\mathbb P}(1).
$$
Lebesgue's Dominated Convergence Theorem implies
$$
\lim _{\theta\rightarrow\theta_0} \mathbb{E}\{ H_{1n}( S_C; \theta )\} = \mathbb{E}\{ H_{1n}( S_C; \theta_0 )\} = A(\theta_0)
$$
and
$$
\lim _{\theta\rightarrow\theta_0} \mathbb{E}\{ H_{2n}( S_C; \theta )\} = \mathbb{E}\{ H_{2n}( S_C; \theta_0 )\} = 0.
$$
Gathering facts, deduce
\begin{equation}\label{hess_neg}
\left\|  \frac{\partial^2 \widehat \ell^*  (\overline \theta) }{\partial \theta \partial \theta ^\top } -   \frac{\partial^2 \widehat \ell^*  ( \theta_0) }{\partial \theta \partial \theta ^\top } \right\| = o_{\mathbb P}(1),
\end{equation}
and this completes the justification of the representation of $\widehat \theta - \theta_0$. The convergence in law of $\sqrt{n}(\widehat \theta - \theta_0) $  is a direct consequence of this representation.

An alternative way to obtain (\ref{hess_neg}) is to require that condition (\ref{suff_cdt_ulln}) holds true and impose more regularity on the regression functions in the cure regression model. More precisely, assume that, in addition to (\ref{suff_cdt_ulln}), there exists an integrable function $C(X)$ and a constant  $a>0$ such that, for any $0\leq k\leq l\leq p$ and any $\theta,\theta^\prime\in\Theta$,
$$
\omega(x) \left|\frac{\partial^2 \pi}{\partial \theta_{(k)} \partial \theta_{(l)} }(x;\theta)  - \frac{\partial^2 \pi}{\partial \theta_{(k)} \partial \theta_{(l)} }(x;\theta^\prime ) \right|\leq C(x) \| \theta - \theta^\prime\|^a,
$$
and
$$
\omega(x)  \left| \left( \frac{\partial \pi}{\partial \theta_{(k)} }  \frac{\partial \pi}{\partial \theta_{(l)}} \right)  (x;\theta) -  \left( \frac{\partial \pi}{\partial \theta_{(k)} }  \frac{\partial \pi}{\partial \theta_{(l)}} \right)  (x;\theta^\prime )  \right|\leq C(x) \| \theta - \theta^\prime\|^a.
$$
Then property (\ref{hess_neg}) follows using arguments as in the proof of Lemma \ref{suff_cdt_ulln}.
\end{proof}

\quad

\begin{proof}[Proof of Theorem \ref{oracle}]
We follow along the lines of the proof of Theorem 4 in \citet{zou:2006}. First consider the asymptotic normality part. Let $\theta = \theta_0+ \mathbf{u}n^{-1/2}$ and define
$$
\Gamma_n (\mathbf{u}) =  \widehat \ell^*(\theta_0+ \mathbf{u}n^{-1/2}) -  \lambda \sum_{j=1}^p  w_j \left|\theta_{0,(j)} + \mathbf{u}_{(j)} n^{-1/2}\right|.
$$
Let  $\widehat {\mathbf{u}} = \arg\max_{\mathbf{u}}\Gamma_n (\mathbf{u}) ,$ such that  $\widehat {\mathbf{u}} = \sqrt{n} (\widehat\theta_\lambda - \theta_0).$ By Taylor expansion %applied to the log-likelihood,
$$
\Gamma_n (\mathbf{u}) - \Gamma_n (\mathbf{0}) = A_1^{(n)}+ A_2^{(n)}+A_3^{(n)}+A_4^{(n)},
$$
where
$$
A_1^{(n)}=  \frac{1}{\sqrt{n}}\frac{\partial \widehat \ell^*  (\theta_0)}{\partial \theta}^\top \mathbf{u}, \qquad
%$$
%$$
A_2^{(n)}=  \frac{1}{2n} \mathbf{u}^\top \frac{\partial^2 \widehat \ell^*  (\theta_0) }{\partial \theta \partial \theta ^\top }  \mathbf{u}, \qquad
$$
$$
A_3^{(n)}= \frac{1}{2n} \mathbf{u}^\top \left[  \frac{\partial^2 \widehat \ell^*  (\overline \theta) }{\partial \theta \partial \theta ^\top } -   \frac{\partial^2 \widehat \ell^*  ( \theta_0) }{\partial \theta \partial \theta ^\top } \right]  \mathbf{u},
$$
and
$$
A_4^{(n)} = -  \frac{\lambda}{\sqrt{n}} \sum_{j=1}^p  w_j \sqrt{n} \left( \left|\theta_{0,(j)} + \mathbf{u}_{(j)} n^{-1/2}\right| - \left|\theta_{0,(j)} \right|\right).
$$
The behaviour of  $A_1^{(n)}$, $A_2^{(n)}$ and $A_3^{(n)}$ can be derived from the proof of Theorem \ref{prop_tcl}. On the other hand, for $A_4^{(n)}$, by the same arguments as in the proof of Theorem 4 in \citet{zou:2006}, we have
$$
\frac{\lambda}{\sqrt{n}}  w_j \sqrt{n} \left( \left|\theta_{0,(j)} + \mathbf{u}_{(j)} n^{-1/2}\right| - \left|\theta_{0,(j)} \right|\right) \; \rightarrow  \; \left\{
  \begin{array}{lll}
      0 &  \text{if} &  \theta_{0,(j)} \neq 0  \\
      0 &  \text{if} &   \theta_{0,(j)} = 0 \text{ and } \mathbf{u}_{(j)} = 0 \\
-\infty &  \text{if} &  \theta_{0,(j)} = 0 \text{ and } \mathbf{u}_{(j)} \neq 0
  \end{array}
\right.,
$$
in probability. To summarize, for every $\mathbf{u}$, $\Gamma_n (\mathbf{u}) - \Gamma_n (\mathbf{0})$ converges in probability to the function
$$
H_n(\mathbf{u}) = \left\{
  \begin{array}{lll}
        W_{n,\mathcal{A}}\mathbf{u}_{\mathcal{A}} -  \mathbf{u}_{\mathcal{A}} ^\top A_{\mathcal{A}}(\theta_0) \mathbf{u}_{\mathcal{A}} / 2&  \text{~~~if} &   \mathbf{u}_{(j)} = 0  \; \forall j \not \in \mathcal{A} \\
-\infty &  \text{~} &   \text{ otherwise }
  \end{array}
\right. ,
$$
where
$$
W_{n,\mathcal{A}} = \frac{1}{\sqrt{n}} \sum_{i=1}^n  \left\{ \mu_\mathcal{A}(Y_i, \Delta_i, X_i;\theta_{0})+ \mu_{\mathcal{A},C}^\varphi (Y_i,\Delta_i, X_i;\theta_{0}) \right\}  .
$$
The iid representation, and hence the asymptotic normality for $\mathbf{u}_{\mathcal{A}}$, follows. Moreover, we deduce that $\widehat {\mathbf{u}}_{\mathcal{A}^c}\rightarrow 0,$ in probability.

Next, we investigate the consistency part. It remains to show that $\forall j^\prime \not\in \mathcal{A}$, we have $\Pr(j^\prime \in \mathcal{A}_n) \rightarrow 0. $ Let us fix arbitrarily $j^\prime \in \mathcal{A}^c$ and consider the event $\{j ^\prime \in \mathcal{A}_n \}$. Let $\mathbf{e}_{j^\prime}\in\mathbb{R}^{p+1}$ be the vector with zero components, except the $j^\prime$th component that is equal to 1.
By the Karush-Kuhn-Tucker optimality conditions, we necessarily have
$$
 \mathbf{e}_{j^\prime} ^\top \; \frac{\partial  \widehat \ell^*(\widehat \theta_\lambda)}{\partial \theta} = \lambda  w_{j\prime }.
$$
Thus
$$
\Pr(j^\prime \in \mathcal{A}_n)\leq  \Pr\left(\mathbf{e}_{j^\prime} ^\top \;   \frac{\partial  \widehat \ell^*(\widehat \theta_\lambda)}{\partial \theta} = \lambda w_{j\prime } \right) .
$$
Next, by Taylor expansion, we {\red can} decompose
$$
\frac{1}{\sqrt{n}}   \mathbf{e}_{j^\prime} ^\top \; \frac{\partial  \widehat \ell^*(\widehat \theta_\lambda)}{\partial \theta} = B_1^{(n)}+ B_2^{(n)}+B_3^{(n)},
$$
with
$$
B_1^{(n)} = \frac{1}{\sqrt{n}} \mathbf{e}_{j^\prime} ^\top \; \frac{\partial  \widehat \ell^*( \theta_0 )}{\partial \theta} , \qquad
%$$
%$$
B_2^{(n)} = \frac{1}{n} \mathbf{e}_{j^\prime} ^\top \; \frac{\partial^2 \widehat \ell^*( \theta_0 )}{\partial \theta\partial \theta^\top }  \sqrt{n} ( \widehat\theta_{\lambda} - \theta_0),
$$
and
$$
B_3^{(n)} = \frac{1}{n} \mathbf{e}_{j^\prime} ^\top \left[ \frac{\partial^2 \widehat \ell^*( \overline \theta_\lambda )}{\partial \theta\partial \theta^\top }  - \frac{\partial^2 \widehat \ell^*( \theta_0 )}{\partial \theta\partial \theta^\top } \right] \sqrt{n} ( \widehat\theta_{\lambda} - \theta_0) ,
$$
with $ \overline \theta_\lambda$ between $\widehat \theta_\lambda$ and $\theta_0$. By arguments that we have already used, we obtain
$$
B_1^{(n)} = O_{\mathbb{P}}(1),\quad B_2^{(n)} = O_{\mathbb{P}}(1) \quad \text{and} \quad B_3^{(n)} = o_{\mathbb{P}}(1).
$$
Meanwhile, since $\sqrt{n}(\widehat \theta - \theta_0)= O_{\mathbb{P}}(1)$ and $\lambda n^{(\gamma-1)/2} \rightarrow \infty$,
$$
\frac{1}{\sqrt{n}} \; \lambda  w_{j\prime }  =  \frac{\lambda  }{\sqrt{n}} n^{\gamma/2} \frac{1}{\left|\sqrt{n} \widehat\theta_{(j^\prime)} \right|^\gamma} \rightarrow \infty,\qquad \text{in probability}.
$$
Thus, $\Pr(j ^\prime \in \mathcal{A}_n)\rightarrow 0$ {\red which completes the proof}. \end{proof}

\subsection{Uniform convergence and iid representations: examples\label{sec:unifiid}}

In this section we review several approaches for estimating the conditional law of the censoring time used for IPCW in such way that Assumption \ref{ulln_ass} and \ref{uclt_ass2} hold true.  Namely, we consider nonparametric estimators, such as the Kaplan-Meier and the conditional Kaplan-Meier estimators, and
existing estimators in semiparametric models such as proportional hazards, proportional odds and transformation models. For all these models, we propose a guideline to derive iid representations under our assumptions from the existing asymptotic results.

Without loss of generality, we consider the case of a real-valued   function $\varphi(X)$. Recall that we are interested in functions $\varphi$ such that $\omega(x) = 0$ implies $\varphi(x)=0.$ \color{black}In particular this implies
$$
 \varphi (x)= \ind(\omega(x)>0)  \varphi (x), \quad \forall x.
$$\color{black}
For vector-valued functions it suffices to apply the results presented below for each component. For simplicity, we assume that the sample space $\mathcal{X}$ is a subset of a finite-dimensional space.  Moreover, we assume  $\Pr(T=C)=0$ and there exists $\tau\in\mathbb{R}$ such that $\Pr(T_0 > \tau) = 0 $ and $\inf_{x\in\mathcal{X},\;\omega(x)>0 }S_C(\tau \mid x) >0. $

%
%Below we consider the Kaplan-Meier and conditional Kaplan-Meier estimators. Details can be found in \citet{stute1993}, \citet{stute96}, and \citet{lopez2011}, but, for convenience, we adapt their results the notation of our article. Note that, in the following sections, we refer to Assumptions \ref{ulln_ass} and \ref{uclt_ass2}, and condition (\ref{suff_c1}) --- these are all contained in the main paper.

\subsubsection{Kaplan-Meier inverse probability weighting}
When the law of the censoring variable $C$ does not depend on the covariates, the survivor function $S_C$ {\red can} be estimated by the Kaplan-Meier estimator
\begin{multline*}
 \widehat{S}_C \left( y\right) =\prod\limits_{j:Y_{j}\leq y}\left(1 - \frac{1}{n \{1- \widehat H (Y_j-)\}}\right) ^{1-\Delta _{j}}, \\ \text{with }\quad n \{1- \widehat H (y-)\}  = \sum_{j=1}^n \ind(Y_j \geq t).
\end{multline*}
Here $\widehat H (y-)$ is {\red the} left-sided limit of the estimate of marginal distribution function of the observed lifetimes $H (y)=\Pr(Y\leq y)$, $y\in\mathbb{R}$. Then the uniform convergence condition (\ref{suff_c1}), and thus Assumption \ref{ulln_ass}, is guaranteed by the uniform law of large numbers of \citet{stute1993}.

For Assumption  \ref{uclt_ass2}, one can use the  iid representation for Kaplan-Meier integrals, as stated in Theorem 1.1 of \citet{stute96}.
More precisely,  for each  $y\in\mathbb{R},$ define
$$
\gamma_0(y) =  S_C(y-)^{-1},  \qquad \gamma_1 (y) = \frac{1}{1-H(y)} \int \ind(y < t  ) \gamma_0(t) \varphi (x)  H^{11} (dx,dt) ,
$$
and
$$
 \gamma_2(y) =  \int \int \frac{\ind(s < y, s < t) \gamma_0(t)   }{\{1-H(s)\}^2} \varphi (x)  H_0(ds)  H^{11} (dx,dt)
$$
with $H^{11} (x,y) = \Pr(X\leq x, Y\leq y, \Delta=1).$ Note that $H^{11} (x,y)=H^{11} (x,\tau)$, $\forall y> \tau$.  Then Assumption  \ref{uclt_ass2} holds with
\begin{equation*}%\label{mu_C_KM}
\mu_C^\varphi (Y_i,\Delta_i,X_i) =   (1-\Delta_i)\gamma_1(Y_i) - \gamma_2(Y_i) ,\qquad 1\leq i\leq n.
\end{equation*}

\subsubsection{Conditional Kaplan-Meier inverse probability weighting}

The conditional Kaplan-Meier estimator, also called  Beran estimator \citep{beran:1981} is defined as
\begin{equation*}%\label{beran_C}
\widehat S_C \left(y \mid x\right) = \large{\prod_{Y_i\leq y }} \left( 1 - \frac{\widehat w_{in}(x)}{\sum_{j=1}^n \widehat w_{jn}(x)\ind(Y_j \geq Y_i) }\right)^{1-\Delta_i},
\end{equation*}
where
$$
\widehat w_{in}(x) = \frac{K((X_i - x))/b_n))}{\sum_{j=1}^n  K(( X_j - x))/b_n)},\qquad x\in\mathcal{X} .
$$
Here, $b_n$ is a bandwidth sequence and $ K(\cdot)$ is a multivariate kernel function. The uniform law of large numbers {\red for} this conditional Kaplan-Meier estimator was established in Corollary 2.1 of \citet{dabrowska89} under some regularity conditions on the density of $X$ and the functions $x\mapsto \Pr(Y\leq y ,\Delta=j\mid X=x),$ $j=0,1$, $y\leq \tau$. To apply that corollary in our framework, one has to define a set on which the density of the covariate vector stays away from zero and to take the weight function  $\omega(\cdot)$ in the definition of our estimators $\widehat \theta$ and $\widehat \theta_{\lambda}$ equal to zero outside this set.

The iid representation {\red for} Kaplan-Meier integrals was extended to conditional Kaplan-Meier integrals, see \citet{lopez2011}. However, {\red this} purely nonparametric approach suffers from the {\red curse} of dimensionality when the sample space $\mathcal{X}$ is multidimensional. Stronger regularity assumptions and high-order kernels are needed in such cases. For simplicity, following \citet{lopez2011}, assume that $C\perp X\mid Z$ where $Z=g(X)\in\mathbb{R}$ with $g(\cdot)$ a given function. For instance, $g(X)$ \color{black} could be a component of $X$, or a given linear combination of components of $X$.  The case where $g(\cdot)$ is known up to a finite-dimensional parameter {\red which} has to be estimated could be also considered, but would introduce an additional term in the iid representation in Assumption \ref{uclt_ass2} {\red which} takes into account the estimation of $g(\cdot)$. For simplicity, herein we assume that $g(\cdot)$ is given. Moreover, for some small $\delta>0$, the weight function $\omega(\cdot)$
vanishes outside the set $\mathcal{X}_\delta = \{x\in\mathcal{X}: f_g (g(x)) \geq  \delta+ \inf_{z}f_g (z)\}$. Here, $f_z$ denotes the density of $Z=g(X),$ {\red which we assume exists and satisfies} some differentiability conditions.
Then, {\red the}  Beran estimator with $Z_i=g (X_i),$ $1\leq i \leq n,$
 is defined as
\begin{equation*}%\label{beran_C}
\widehat S_C \left(y \mid z\right) = \large{\prod_{Y_i\leq y }} \left( 1 - \frac{\widehat w_{in}(z)}{\sum_{j=1}^n \widehat w_{jn}(z)\ind(Y_j \geq Y_i) }\right)^{1-\Delta_i}, \qquad z\in\mathbb{R},
\end{equation*}
where
$$
\widehat w_{in}(z) = \frac{K((Z_i - z))/a_n))}{\sum_{j=1}^n  K(( Z_j - z))/a_n)}.
$$
Here, $a_n$ is a bandwidth sequence converging to zero as $n$ tends to infinity, and $ K(\cdot)$ is a univariate kernel function. If the bandwidth $a_n$ satisfies $\log(n)n^{-1}a_n^{-3} \rightarrow 0$ and $na_n^4\rightarrow 0,$
\begin{multline*}
\frac{1}{n} \sum_{1\leq i\leq n} \left[ \frac{\Delta_i }  {\widehat S_C(Y_i- \mid  Z_i)}  - \frac{\Delta_i }  {S_C(Y_i-\mid  Z_i)} \right] \varphi (X_i) \\= \frac{1}{n} \sum_{1\leq i\leq n} \mu_C^\varphi (Y_i,\Delta_i, X_i) + o_{\mathbb{P}}(n^{-1/2}) ,
\end{multline*}
where
\begin{equation}\label{mu_C_KMc}
\mu_C^\varphi  (Y_i,\Delta_i,X_i;\theta_0) =   (1-\Delta_i)\gamma_1(Y_i, Z_i) - \gamma_2(Y_i,Z_i) ,\qquad 1\leq i\leq n,
\end{equation}
with
\begin{align*}
\gamma_0(y,z) &=  S_C(y\mid z )^{-1},  \\
\gamma_1 (y,z) &= \frac{1}{1-H(y\mid z)} \int \ind(y < t)\gamma_0(t,z)  \varphi(x)  H^{11} (dx,dt\mid z) ,
\end{align*}
and
$$
\gamma_2(y,z) =  \int \int \frac{\ind(s < y, s < t ) \gamma_0(t,z)   }{\{1-H(s\mid z)\}^2} \varphi(x)  H_0(ds\mid z )  H^{11} (dx,dt\mid z ).
$$
Here, $H_0(y\mid z )=\Pr(Y\leq y, \delta=0\mid g(X)=z)$ and $H^{11} (x,y\mid z) = \Pr(X\leq x, Y\leq y, \delta=1\mid g(X)=z).$

\subsubsection{Semiparametric models}

In this section we present a general method for guaranteeing Assumption \ref{uclt_ass2}. This method {\red can, for example,} be applied to the common semiparametric models used in survival analysis. (For all the models we mention, condition (\ref{suff_c1}) is obviously satisfied, and thus Assumption \ref{ulln_ass}.) {\red Note that what we present here is a fuller version of the sketch which appears in Section \ref{sec:normal} of the main paper, and, thus, there is some repetition which we maintain to improve readability herein.} Assume that the conditions of Lemma \ref{suff_cdt_ulln} hold true and $\varphi(\cdot)$ is bounded. We can write
\begin{multline*}
\frac{1}{n} \sum_{1\leq i\leq n} \left[ \frac{\Delta_i }  {\widehat S_C(Y_i- \mid  X_i)}  - \frac{\Delta_i }  {S_C(Y_i-\mid  X_i)} \right] \varphi (X_i) \\=   \frac{1+ o_{\mathbb{P}}(1)}{n} \sum_{1\leq i\leq n} \left\{q(Y_i,\Delta_i,X_i;  S_C) - q(Y_i,\Delta_i,X_i; \widehat S_C)  \right\}
\end{multline*}
where for any function $S$ depending on $y$ and $x$ that is c\`adl\`ag in $y$,
$$
q(t,d,x; S) = d \frac{\varphi (x)}{S^2_C(t-\mid x)} \times  S(t-, x).
$$
Then Assumption \ref{uclt_ass2} {\red can} be guaranteed in two steps. First,  use the equicontinuity of Donsker classes and  transform the sum with respect to $1\leq i\leq n$ {\red to} an expectation with respect to a generic triplet $(Y,\Delta,X)$ given the sample. Next, use the iid representation of $\widehat S_C(y-\mid  x)$ which can be derived in common semiparametric models. For the first step, in the following, we introduce a general class of functions
for which we prove the Donsker property. For several common semiparametric models and estimators, it can be shown that  $q(\cdot,\cdot,\cdot; S_C)$ and $q(\cdot,\cdot,\cdot; \widehat S_C)$ belong to this class with probability tending to 1. For other models, it may be necessary to define alternative Donsker classes to be used in the first step of the method we propose herein.

Let us suppose that $X$ is composed of continuous and discrete components, that is $X = (X_d^\top,X_c^\top )^\top\in\mathcal{X}_d\times \mathcal{X}_c\subset \mathbb{R}^{p_d}\times \mathbb{R}^{p_c}.$ Then, each vector $x$ in the support of $X$ {\red can} be split in the subvectors  $x_d\in \mathcal{X}_d$ and $x_c\in \mathcal{X}_c.$ For simplicity, assume that the support $\mathcal{X}_d$ %of the discrete subvector $X_d$
 is finite.

Let $BV_M [0,\tau]$ be the set of all real-valued c\`adl\`ag functions defined on $[0,\tau]$ with the total variation bounded by $M$. Let $\mathcal{C}_M^\alpha (\mathcal{X}_c) $ {\red be} the set of all continuous functions $f:\mathcal{X}_c \rightarrow \mathbb{R}$ with $\|f\|_\alpha \leq M,$ where  $\|\cdot\|_\alpha$ is the usual uniform norm defined on the class of functions with uniformly bounded partial derivatives up to order $\underline{\alpha}$ (the greatest integer smaller than $\alpha$) and whose highest partial derivatives are Lipschitz of order $\alpha - \underline{\alpha}$. See chapter 2.7.1 in \citet{vaartwellner96book}.
We follow \citet{lopez2011} and, for $\alpha,M >0$,  define the following class of functions defined on $[0,\tau]\times \mathcal{X}$~:
\begin{multline*}
\mathcal{G} = \left\{ (t,x)\mapsto h(t,x):  \forall x, h(\cdot,x) \text{ and } D^{\mathbf{k}} h(\cdot,x)\in BV_M [0,\tau], \right.\\ \left. \text{ and } \forall t\leq \tau, \forall x_d\in \mathcal{X}_d,  h(t,(x_d,\cdot)) \in\mathcal{C}_M^\alpha (\mathcal{X}_c)   \right\},
\end{multline*}
where for any vector $\mathbf{k} = (k_1,\ldots,k_{p_c})$ of $p_c$ integers, $D^{\mathbf{k}} $ is the differential operator
$$
D^{\mathbf{k}} = \frac{\partial ^{k_1+\cdots+k_{p_c}}}{\partial x_{c,1}\cdots \partial x_{c,p_c}}.
$$
The next result shows that $\mathcal{G}$ is a general Donsker class. For this it suffices to show that $\mathcal{G}$ has finite bracketing integral, which here is tantamount to
$
\int_{0}^1 \sqrt{\log N_{[\; ]}(\varepsilon, \mathcal{G}, L_2)} \; d\varepsilon<\infty,
$
where $N_{[\; ]}(\varepsilon, \mathcal{G}, L_2)$ {\red is} the $\varepsilon-$bracketing number of the class $\mathcal{G}$ with respect to the $L_2-$norm.

\begin{lemma}\label{lem_donsk}
Assume that $\mathcal{X}_c$ is a bounded, open and convex subset of $\mathbb{R}^{p_c}$ and $\mathcal{X}_d$ is finite. If $M<\infty$ and $\alpha > p_c$,
$\mathcal{G}$ is a  Donsker class.
\end{lemma}
\begin{proof}[Proof of Lemma \ref{lem_donsk}]
We provide a sketch proof. The idea is to show that
\begin{equation}\label{b22}
\log N_{[\; ]}(\varepsilon, \mathcal{G}, L_\infty)\leq C (1/\varepsilon)^{2p_c/\alpha} \log(1/\varepsilon) ,
\end{equation}
for some constant $C$. Since $\int_0^1 (1/\varepsilon)^{p_c/\alpha}\log^{1/2}(1/\varepsilon)  d\varepsilon <\infty$ when $p_c/\alpha <1$, this will guarantee that $\mathcal{G}$ has finite bracketing integral.
With a finite set $\mathcal{X}_d$, by the permanence of the Donsker property (see Theorem 2.10.6 in \citet{vaartwellner96book}), it suffices to show property (\ref{lem_donsk}) separately for each value  $x_d\in\mathcal{X}_d.$ In other words, it suffices to consider that $x=x_c$.

First, for any $t\in[0,\tau],$ let $\mathcal{G}_t = \{ x\mapsto h(t,x) : h\in\mathcal{G} \}.$ Then, by Corollary 2.6.2 of \citet{vaartwellner96book}, we have {\red that} $\log N_{[\; ]}(\varepsilon, \mathcal{G}_t   , L_2)\leq C_1 \varepsilon^{-p_c/\alpha},$ for each $t\in[0,\tau]$, {\red where $C_1$ is} a constant depending only on $M$, $\alpha,$ $\text{diam}(\mathcal{X}_c)$ and $p_c$ (and {\red is } independent of $t$). Meanwhile, $\log N_{[\; ]}(\varepsilon, BV_M [0,\tau], L_\infty) \leq C_2 (1/\varepsilon) \log(1/\varepsilon),$ where $C_2$ is a constant depending only on $M$ and $\tau$. This is a simple bound for the bracketing number, {\red and can} be obtained by {\red constructing brackets which} are piecewise constant on a regular grid. See also the beginning of the proof of Theorem 2.7.5 in \citet{vaartwellner96book}. In the following, without loss of generality, we consider $M=1$. To prove (\ref{b22}), it suffices to reconsider the proof of Theorem 2.7.1 (and Corollary 2.7.2) of \citet{vaartwellner96book}. More precisely, in the the proof of Theorem 2.7.1,  replace the entries of the matrices $Af$ by functions of $t$ {\red with} total variation bounded by $M$. In \citet{vaartwellner96book}, the entries of $Af$ are defined for $f$ ranging over $\mathcal{C}_1^\alpha (\mathcal{X}_c)$ and having the rows built using the values of the partial derivatives of $f$, up to order equal to the greatest integer strictly smaller than $\alpha$,  discretized on a grid with the mesh controlled by $\varepsilon$.  In our case, such entries {\red may} depend also on $t$ and be taken as a set of brackets covering $BV_M [0,\tau]$. Deduce a set of brackets that cover $\mathcal{G}$ and the cardinality of the set of brackets is of order
$$
\exp\left(\{\varepsilon^{-1/\alpha} +\ldots +\varepsilon^{-p_c/\alpha}  \}\log(1/\varepsilon)  \right)^{\varepsilon^{-p_c/\alpha}}.
$$
Taking the logarithm of the above, we recover the bound $(1/\varepsilon)^{2p_c/\alpha}\log(1/\varepsilon) $ for the order of the bracketing entropy of $\mathcal{G}$. \end{proof}

\quad

As a first step for guaranteeing Assumption \ref{uclt_ass2}, using the properties of the model and estimator considered, one {\red can} check that
\begin{equation}\label{donsk}
S_C\in  \mathcal{G} \quad \text{and} \quad \Pr (\widehat S_C\in\mathcal{G})\rightarrow 1
\end{equation}
in order to deduce, using the permanence properties of Donsker classes,  that $q(\cdot,\cdot,\cdot; S_C)$ and $q(\cdot,\cdot,\cdot; \widehat S_C)$ belong to a Donsker class with probability tending to 1. By the asymptotic equicontinuity of Donsker classes (see \citet{vaartwellner96book}), and our Lemma \ref{core_id},
\begin{multline*}%\label{step1}
\frac{1}{n} \sum_{1\leq i\leq n} \left[ \frac{\Delta_i }  {\widehat S_C(Y_i- \mid  X_i)}  - \frac{\Delta_i }  {S_C(Y_i-\mid  X_i)} \right] \varphi (X_i) \\ =  \int\int \mathbb{G}_n (t,x)  %\ind(\omega(x)>0)
\varphi (x)\{1-\pi (x)\} dF(x,t) + o_{\mathbb{P}}(n^{-1/2}),
\end{multline*}
where
$$
\mathbb{G}_n (t,x) = \frac{S_C(t-\mid x) - \widehat S_C(t-\mid x)}{S_C(t-\mid x)},\qquad t\in [0,\tau], x\in\mathcal{X},
$$
and $F(x,t) = \Pr(X\leq x, T_0\leq t)$.
In the second step, it remains to show that
\begin{equation}\label{iid_mod_semi}
\text{the process } \; \mathbb{G}_n (\cdot,\cdot) \; \text{ admits an iid representation}
\end{equation}
using the properties of the model considered for $S_C$.

In the following we \color{black} study conditions \color{black} (\ref{donsk}) and (\ref{iid_mod_semi}) in several common (semi)parametric models in survival analysis. See \citet{GuoZeng12} for an illuminating survey on  semiparametric models.
In these models, the properties of $\widehat S_C$ are usually derived from the properties of the conditional cumulative hazard function of $C$.  Note that by the Duhamel identity \citep{gill1990} %, if condition (\ref{suff_cdt_ulln}) holds true,
we have
\begin{multline*}%\label{duh_C}
\mathbb{G}_n (t,x) = -\int_{\color{black} (0,t) \color{black}} \frac{\widehat S_C(s-\mid x) }{S_C(s\mid x)} d(\widehat\Lambda_C - \Lambda_C)(s\mid x) \\ = -\{1+o_{\mathbb{P}}(1)\} \int_{\color{black} (0,t) \color{black}} \frac{ S_C(s-\mid x) }{S_C(s\mid x)} d(\widehat\Lambda_C - \Lambda_C)(s\mid x), \qquad t\leq \tau,
\end{multline*}
where $\Lambda_C(\cdot \mid x)$ is the conditional cumulative hazard function of $C$ given $X=x$ and $\widehat \Lambda_C(\cdot \mid x)$ is the estimator of $\Lambda_C(\cdot \mid x)$ in the model for $S_C$. When $S_C(\cdot \mid x)$ is continuous,
 $$
 \mathbb{G}_n (t,x) =  - \left\{ \widehat\Lambda_C (t\mid x) - \Lambda_C(t\mid x) \right\} \{1+o_{\mathbb{P}}(1)\}.
 $$

\textbf{Parametric models.} One {\red can} achieve flexible modeling for $S_C$ using parametric models, such as the Weibull model, where the parameters are replaced by functions of $X$ depending on some unknown vector of coefficients $\beta_C$. Then, in general, all the survivor functions in the model satisfy the regularity conditions defining the class $\mathcal{G}$, so that condition (\ref{donsk}) is automatically met. Condition (\ref{iid_mod_semi}) follows by Taylor expansion and the asymptotic linear expansions of $\widehat \beta_C$, the estimator of $\beta_C$ considered. Such asymptotic linear expansions for $\widehat \beta_C$, based on the so-called \emph{influence functions}, are available for all common estimators $\widehat\beta_C$.% could be derived for all the common models and estimators $\widehat\beta_C$.

\textbf{Cox's proportional hazard model.} In this case
$$
\Lambda_C (t\mid x)= \exp(x^\top \beta_C) \Lambda_{0,C}(t),
$$
where $\Lambda_{0,C}(\cdot) $ is the so-called baseline cumulative hazard function. Clearly, $\Lambda_C (t\mid x)$ belongs to $\mathcal{G}$. Any estimate of $\Lambda_C (t\mid x)$ belongs to $\mathcal{G}$ {\red once} $\Lambda_{0,C}(\cdot) $ is estimated by a function {\red with} total variation  bounded by some suitable $\widetilde M$ depending on $M$, the compact parameter set for $\beta_C$, and the bounded $\mathcal{X}$.
\color{black} Thus, condition (\ref{donsk}) is easily granted. \color{black}
Next, for any consistent estimator $\widehat\beta_C$ and uniformly consistent $\widehat \Lambda_{0,C}(\cdot) $ we {\red can} write
\begin{multline*}
\mathbb{G}_n (t,x) = - \left[  \exp(x^\top \beta_C) \left\{ \widehat\Lambda_{0,C} (t) - \Lambda_{0,C} (t) \right\} + \Lambda_C (t\mid x) (\widehat \beta_C - \beta_C ) \right]\\\times  \{1+o_{\mathbb{P}}(1)\}.
\end{multline*}
In the case where $\widehat \beta_C$ is the maximum partial likelihood estimator, and $ \widehat\Lambda_C$ is the associated Breslow estimator, condition (\ref{iid_mod_semi}) follows by the asymptotic results of \citet{andersen1982}.

\textbf{Transformation model.} We consider the class of transformation models investigated by \citet{ZhengLin06}, where
$$
\Lambda_C (t\mid x) = G\left( \exp(x^\top \beta_C) \Lambda_{0,C}(t) \right),
$$
where $G(\cdot)$ is a given smooth, strictly increasing transformation function with $G(0)=0,$ $G^\prime (0) >0$ and $G(\infty)=\infty$. Box-Cox transformations, and log transformations are possible examples.
\citet{ZhengLin06} extend the partial likelihood idea and the Breslow estimator from Cox's model and
introduce the estimators $\widehat \beta_C$ and $\widehat\Lambda_{0,C}$ for which they derive a Gaussian limit. Since
\begin{multline*}
\widehat \Lambda_C (t\mid x) - \Lambda_C (t\mid x)  = G\left( \exp(x^\top \widehat \beta_C) \widehat \Lambda_{0,C}(t) \right)
- G\left( \exp(x^\top \beta_C)  \Lambda_{0,C}(t) \right)\\
= G^\prime \left(  \exp(x^\top \beta_C)  \Lambda_{0,C}(t) \right) \left\{ \exp(x^\top \widehat \beta_C) \widehat \Lambda_{0,C}(t) -  \exp(x^\top \beta_C)  \Lambda_{0,C}(t)  \right\}  \\ \times  \{1+o_{\mathbb{P}}(1)\},
\end{multline*}
we {\red can} easily adapt the arguments used for the case of Cox's model and, based on definition of $\widehat \beta_C$ and $\widehat \Lambda_{0,C}(\cdot)$ and  the asymptotic results of \citet{ZhengLin06}, guarantee conditions (\ref{donsk}) and (\ref{iid_mod_semi}).

\textbf{Proportional odds model.} In such a model we have
$$
S_C(t\mid x)  = \frac{\exp(-x^\top \beta_C)}{H(t)+ \exp(-x^\top \beta_C)},
$$
with $H(\cdot)$ some c\`adl\`ag function with $H(0)=0$.
See \citet{Murphy97}. Clearly, $S_C\in\mathcal{G}$ and $\widehat S_C\in\mathcal{G},$ with $\widehat S_C(t\mid x)  = \exp(-x^\top \widehat \beta_C)/\{\widehat H(t)+ \exp(-x^\top \widehat \beta_C)\}$, {\red once} $\widehat H(t)\in BV_{\widetilde M }[0,\tau]$ for some suitable $\widetilde M$. For the required iid representation, we {\red can} use the maximum likelihood estimators for $\beta_C$ and $H(\cdot)$, linearize the expression of $S_C$ with respect to the parameters and use the asymptotic representations of $\sqrt{n}(\widehat H - H)$ and $\sqrt{n}(\widehat \beta_C - \beta_C)$ established in the proof of the Theorem 2.2 in  \citet{Murphy97}. That result is established under an additional condition which, in our setup, means that $\Pr(T_0=\tau) = \Pr(T_0 \geq \tau ) >0$. This additional technical constraint is quite usual in the cure regression literature. See, for instance, \citet{Lu2008} and \citet{Fang2005}.

\subsection{Implementation of the alasso\label{sec:imp}}

\subsubsection{Optimization}

A variety of algorithms have been implemented for solving non-differentiable lasso problems, e.g., quadratic programming \citep{tibshirani:1996}, least angle regression (LARS) \citep{efronetal:2004}, and co-ordinate descent \citep{friedmanetal:2007}. However, we prefer the use of a differentiable penalty since standard gradient-based optimization procedures can then be utilized. Therefore, we propose the use of
\begin{equation}
\widehat\ell^*_{\lambda,\epsilon}(\theta) = \widehat\ell^*(\theta) - \lambda \sum_{j=1}^p  w_j a_\epsilon(\theta_{(j)}) \label{penlike}
\end{equation}
where $a_\epsilon(x) = (x^2 + \epsilon^2)^{1/2} - \epsilon$ is an extension of the absolute value function such that $\lim_{\epsilon \rightarrow0}a_\epsilon(x) = |x|$, and which is differentiable for $\epsilon>0$. Clearly,
smaller $\epsilon$ values bring the penalty closer to the alasso, but also bring (\ref{penlike}) closer to being non-differentiable.
In our work, we have found that $\epsilon = 10^{-4}$ works well.

\subsubsection{Tuning parameter selection}

For the purpose of selecting the tuning parameter, $\lambda$, we consider cross-valida\-tion; in particular, we aim to minimize the $k$-fold cross-validation error. Since we have $\expect\{B_i^*(S_C) - \pi_i(\theta)\mid X_i\} = 0$, we may define the error term, $B_i^*(\widehat S_C) - \pi_i(\widehat\theta)$. Then, for a partition $F_1,\ldots,F_K$ of the set $\{1,\ldots,n\}$, the mean-squared error for the $j$th fold, $F_j$ is given by $\sum_{i\in F_j}\{B_i^*(\widehat S_C) - \pi_i(\widehat\theta_\lambda^{-j})\}^2$ where $\widehat\theta_\lambda^{-j}$ is the penalized estimate with the $j$th fold removed. Thus, the $k$-fold cross-validation error is
\begin{align}
\text{CVE}(\lambda) = \frac{1}{k}\sum_{j=1}^{\text{k}}\sum_{i\in F_j}\left\{B_i^*(\widehat S_C) - \pi_i(\widehat\theta_\lambda^{-j})\right\}^2\label{cvmse}
\end{align}
where we use $k=10$ as is standard in practice.
 Minimizing (\ref{cvmse}) with respect to $\lambda$ can be achieved by profiling over a range of $\lambda$ values or by using a one-dimensional optimizer, e.g., golden search; we will define $\lambda^{\text{CVE}}$ to be the minimizer of  (\ref{cvmse}). {\red One might also consider a BIC-type criterion of the form $-2\widehat\ell^*(\widehat \theta_\lambda) + \text{dim}(\widehat \theta_\lambda)\log(n)$. We have also tested this in our simulation study where it produced very similar results to the cross-validation approach; thus, in Section \ref{sec:sim}, we only present cross-validation.}

\subsection{Additional simulation results \label{sec:simadd}}

Our simulations setup, described in Section \ref{sec:simsetup} of the main paper, comprises 24 scenarios, whereas, only the results for the 6 scenarios where $\pi_m=0.4$ and $\rho=0.1$ were reported on {\red in detail}. Therefore, here, all results are displayed in Table \ref{tab:resadd}. The estimates improve with the sample size (as seen in the main paper), disimprove with increased censoring proportion (controlled via $\rho$), and improve with an increased cure proportion, $\pi_m$. The estimates from our approach and \texttt{smcure} are similar, with \texttt{smcure} being slightly more efficient; of course, in the case where $\nu = 2$ (non-PH $S_{T_0}$) the \texttt{smcure} estimates are biased. {\red In Section \ref{sec:simselect}, we described a simulation setup with six covariates. However, our focus in that setting was on the performance of our variable selection procedure. For completeness, Table \ref{tab:resk6} displays the results of estimation in that setting which are quite similar to the case with two covariates (albeit with slightly increased bias and variability as expected).}

\begin{table}[!htbp]
\caption{Average bias and standard error (in brackets) of estimates\label{tab:resadd}}
\centering
%\smallskip
\begin{footnotesize}
\begin{tabular}{cccc@{~~~~}c@{~~}c@{~~}c@{~~~~}c@{~~}c@{~~}c@{~~~~}c@{~~}c@{~~}c}
\hline
&&&& \multicolumn{3}{l}{\hspace{0.4cm}$n=100$} &  \multicolumn{3}{l}{\hspace{0.4cm}$n=300$} &  \multicolumn{3}{l}{\hspace{0.4cm}$n=1000$} \\
Method & $\nu$ & $\pi_m$ & $\rho$ & $\theta_{(0)}$ & $\theta_{(1)}$ & $\theta_{(2)}$ & $\theta_{(0)}$ & $\theta_{(1)}$ & $\theta_{(2)}$ & $\theta_{(0)}$ & $\theta_{(1)}$ & $\theta_{(2)}$ \\[0.1cm]
\hline
&&&&&&&&&&&&\\[-0.3cm]
Our
 & 0 &  0.2  & 0.1 &  -0.20 &   0.14 &   0.12 &  -0.06 &   0.05 &   0.03 &  -0.01 &   0.01 &   0.01 \\
proposal
 &   &       &     & (0.65) & (0.62) & (0.58) & (0.29) & (0.29) & (0.26) & (0.15) & (0.14) & (0.14) \\[0.05cm]
 &   &       & 0.2 &  -0.17 &   0.12 &   0.12 &  -0.06 &   0.05 &   0.05 &  -0.02 &   0.01 &   0.02 \\
 &   &       &     & (0.80) & (0.74) & (0.68) & (0.40) & (0.40) & (0.34) & (0.20) & (0.20) & (0.17) \\[0.1cm]
 &   &  0.4  & 0.1 &  -0.05 &   0.10 &   0.12 &  -0.01 &   0.04 &   0.03 &   0.00 &   0.01 &   0.01 \\
 &   &       &     & (0.36) & (0.49) & (0.47) & (0.18) & (0.24) & (0.22) & (0.10) & (0.12) & (0.12) \\[0.05cm]
 &   &       & 0.2 &  -0.03 &   0.16 &   0.17 &  -0.02 &   0.05 &   0.05 &   0.00 &   0.02 &   0.02 \\
 &   &       &     & (0.51) & (0.74) & (0.66) & (0.29) & (0.37) & (0.32) & (0.13) & (0.18) & (0.16) \\[0.1cm]
 & 2 &  0.2  & 0.1 &  -0.16 &   0.12 &   0.11 &  -0.05 &   0.04 &   0.03 &  -0.01 &   0.01 &   0.01 \\
 &   &       &     & (0.62) & (0.60) & (0.51) & (0.29) & (0.29) & (0.24) & (0.15) & (0.15) & (0.13) \\[0.05cm]
 &   &       & 0.2 &  -0.13 &   0.13 &   0.10 &  -0.09 &   0.07 &   0.06 &  -0.02 &   0.01 &   0.01 \\
 &   &       &     & (0.80) & (0.76) & (0.60) & (0.45) & (0.43) & (0.32) & (0.21) & (0.20) & (0.15) \\[0.1cm]
 &   &  0.4  & 0.1 &  -0.05 &   0.13 &   0.11 &  -0.02 &   0.04 &   0.02 &   0.00 &   0.01 &   0.01 \\
 &   &       &     & (0.37) & (0.53) & (0.44) & (0.19) & (0.25) & (0.21) & (0.10) & (0.13) & (0.11) \\[0.05cm]
 &   &       & 0.2 &  -0.03 &   0.15 &   0.15 &  -0.02 &   0.06 &   0.05 &  -0.01 &   0.02 &   0.01 \\
 &   &       &     & (0.55) & (0.77) & (0.60) & (0.30) & (0.41) & (0.27) & (0.15) & (0.20) & (0.13) \\[0.2cm]
\texttt{smcure}
 & 0 &  0.2  & 0.1 &  -0.20 &   0.12 &   0.08 &  -0.05 &   0.04 &   0.02 &  -0.02 &   0.01 &   0.01 \\
 &   &       &     & (0.54) & (0.48) & (0.46) & (0.25) & (0.23) & (0.22) & (0.14) & (0.12) & (0.12) \\[0.05cm]
 &   &       & 0.2 &  -0.40 &   0.20 &   0.14 &  -0.09 &   0.04 &   0.05 &  -0.03 &   0.02 &   0.02 \\
 &   &       &     & (0.93) & (0.72) & (0.68) & (0.34) & (0.30) & (0.27) & (0.17) & (0.15) & (0.14) \\[0.1cm]
 &   &  0.4  & 0.1 &  -0.08 &   0.09 &   0.09 &  -0.02 &   0.03 &   0.02 &  -0.01 &   0.01 &   0.01 \\
 &   &       &     & (0.33) & (0.40) & (0.38) & (0.17) & (0.20) & (0.19) & (0.09) & (0.11) & (0.10) \\[0.05cm]
 &   &       & 0.2 &  -0.25 &   0.20 &   0.12 &  -0.06 &   0.05 &   0.03 &  -0.01 &   0.02 &   0.01 \\
 &   &       &     & (0.64) & (0.68) & (0.61) & (0.26) & (0.27) & (0.25) & (0.12) & (0.14) & (0.13) \\[0.1cm]
 & 2 &  0.2  & 0.1 &  -0.16 &   0.11 &  -0.03 &  -0.03 &   0.02 &  -0.07 &   0.00 &   0.00 &  -0.09 \\
 &   &       &     & (0.54) & (0.48) & (0.42) & (0.25) & (0.23) & (0.20) & (0.13) & (0.12) & (0.11) \\[0.05cm]
 &   &       & 0.2 &  -0.18 &   0.12 &  -0.10 &  -0.02 &   0.02 &  -0.14 &   0.04 &  -0.02 &  -0.16 \\
 &   &       &     & (0.73) & (0.59) & (0.48) & (0.33) & (0.27) & (0.23) & (0.17) & (0.14) & (0.11) \\[0.1cm]
 &   &  0.4  & 0.1 &  -0.12 &   0.10 &  -0.09 &  -0.08 &   0.04 &  -0.13 &  -0.06 &   0.01 &  -0.13 \\
 &   &       &     & (0.35) & (0.40) & (0.33) & (0.19) & (0.21) & (0.17) & (0.10) & (0.11) & (0.09) \\[0.05cm]
 &   &       & 0.2 &  -0.22 &   0.15 &  -0.17 &  -0.12 &   0.05 &  -0.20 &  -0.09 &   0.01 &  -0.20 \\
 &   &       &     & (0.55) & (0.57) & (0.42) & (0.28) & (0.27) & (0.20) & (0.14) & (0.13) & (0.10) \\[0.05cm]
\hline
% &&\\[-0.2cm]
\end{tabular}
\end{footnotesize}
\end{table}

{\red As suggested by an anonymous reviewer, it is also of interest to consider simulation setups with binary covariates. Thus, we considered two further setups: one where $X_{(1)}$ and $X_{(2)}$ are both binary, and another where $X_{(1)}$ is continuous while $X_{(2)}$ is binary. The results, shown in Table \ref{tab:res2}, are numerically very close to those of Table \ref{tab:resadd}. We also considered these additional setups with six covariates, but, to avoid repetition, we do not present the results as they are very close to those of Table \ref{tab:resk6}.}

\begin{table}[!htbp]
\caption{{\red {\bf Absolute} average bias and standard error (in brackets) of estimates: Six covariates}\label{tab:resk6}}
\centering
%\smallskip
{\red
\begin{footnotesize}
\begin{tabular}{cccc@{~~~~}c@{~~}c@{~~}c@{~~~~}c@{~~}c@{~~}c@{~~~~}c@{~~}c@{~~}c}
\hline
&&&& \multicolumn{3}{l}{\hspace{0.4cm}$n=100$} &  \multicolumn{3}{l}{\hspace{0.4cm}$n=300$} &  \multicolumn{3}{l}{\hspace{0.4cm}$n=1000$} \\
Method & $\nu$ & $\pi_m$ & $\rho$ & $\theta_{(0)}$ & $\theta_{(1\text{-}2)}$ & $\theta_{(3\text{-}6)}$ & $\theta_{(0)}$ & $\theta_{(1\text{-}2)}$ & $\theta_{(3\text{-}6)}$ & $\theta_{(0)}$ & $\theta_{(1\text{-}2)}$ & $\theta_{(3\text{-}6)}$ \\[0.1cm]
\hline
&&&&&&&&&&&&\\[-0.3cm]
Our
 & 0 &  0.2  & 0.1 &   0.46 &   0.27 &   0.01 &   0.13 &   0.08 &   0.00 &   0.04 &   0.02 &   0.00 \\
proposal
 &   &       &     & (0.92) & (0.75) & (0.55) & (0.33) & (0.30) & (0.24) & (0.16) & (0.15) & (0.12) \\[0.05cm]
 &   &       & 0.2 &   0.37 &   0.24 &   0.01 &   0.21 &   0.12 &   0.01 &   0.05 &   0.03 &   0.00 \\
 &   &       &     & (1.00) & (0.85) & (0.67) & (0.56) & (0.46) & (0.33) & (0.22) & (0.20) & (0.16) \\[0.1cm]
 &   &  0.4  & 0.1 &   0.11 &   0.23 &   0.01 &   0.03 &   0.07 &   0.00 &   0.01 &   0.02 &   0.00 \\
 &   &       &     & (0.47) & (0.60) & (0.43) & (0.20) & (0.26) & (0.20) & (0.10) & (0.12) & (0.10) \\[0.05cm]
 &   &       & 0.2 &   0.05 &   0.27 &   0.01 &   0.05 &   0.13 &   0.01 &   0.01 &   0.04 &   0.00 \\
 &   &       &     & (0.62) & (0.84) & (0.59) & (0.32) & (0.43) & (0.30) & (0.15) & (0.20) & (0.15) \\[0.1cm]
 & 2 &  0.2  & 0.1 &   0.45 &   0.25 &   0.01 &   0.15 &   0.09 &   0.01 &   0.03 &   0.02 &   0.00 \\
 &   &       &     & (0.89) & (0.74) & (0.55) & (0.36) & (0.32) & (0.24) & (0.16) & (0.15) & (0.12) \\[0.05cm]
 &   &       & 0.2 &   0.24 &   0.14 &   0.02 &   0.23 &   0.14 &   0.01 &   0.08 &   0.05 &   0.00 \\
 &   &       &     & (0.97) & (0.81) & (0.67) & (0.65) & (0.51) & (0.37) & (0.30) & (0.25) & (0.18) \\[0.1cm]
 &   &  0.4  & 0.1 &   0.12 &   0.25 &   0.02 &   0.03 &   0.07 &   0.01 &   0.01 &   0.02 &   0.00 \\
 &   &       &     & (0.44) & (0.61) & (0.44) & (0.21) & (0.27) & (0.20) & (0.10) & (0.13) & (0.10) \\[0.05cm]
 &   &       & 0.2 &   0.05 &   0.24 &   0.02 &   0.04 &   0.11 &   0.01 &   0.02 &   0.04 &   0.00 \\
 &   &       &     & (0.61) & (0.85) & (0.62) & (0.38) & (0.47) & (0.33) & (0.20) & (0.24) & (0.17) \\[0.05cm]
\hline
&&&&&&&&&&&&\\[-0.3cm]
\multicolumn{13}{p{0.92\textwidth}}{\footnotesize $\theta_{(1\text{-}2)}$ pools the results for $\theta_{(1)}$ and $\theta_{(2)}$ which are numerically very close, and, similarly, $\theta_{(3\text{-}6)}$ pools the results for $\theta_{(3)}$, $\theta_{(4)}$, $\theta_{(5)}$, and $\theta_{(6)}$. {\bf Absolute} bias is shown here (which differs from the other tables) so that the pooling of results can be carried out.}
\end{tabular}
\end{footnotesize}
}
\end{table}
\begin{table}[!htbp]
\caption{{\red Average bias and standard error (in brackets) of estimates: Other covariate types}\label{tab:res2}}
\centering
%\smallskip
{\red
\begin{footnotesize}
\begin{tabular}{cccc@{~~~~}c@{~~}c@{~~}c@{~~~~}c@{~~}c@{~~}c@{~~~~}c@{~~}c@{~~}c}
\hline
&&&& \multicolumn{3}{l}{\hspace{0.4cm}$n=100$} &  \multicolumn{3}{l}{\hspace{0.4cm}$n=300$} &  \multicolumn{3}{l}{\hspace{0.4cm}$n=1000$} \\
Covariates & $\nu$ & $\pi_m$ & $\rho$ & $\theta_{(0)}$ & $\theta_{(1)}$ & $\theta_{(2)}$ & $\theta_{(0)}$ & $\theta_{(1)}$ & $\theta_{(2)}$ & $\theta_{(0)}$ & $\theta_{(1)}$ & $\theta_{(2)}$ \\[0.1cm]
\hline
&&&&&&&&&&&&\\[-0.3cm]
Binary
 & 0 &  0.2  & 0.1 &  -0.20 &    0.10 &    0.10 &   -0.07 &    0.04 &    0.03 &   -0.02 &    0.01 &    0.01 \\
 &   &       &     & (0.61) &  (0.52) &  (0.50) &  (0.30) &  (0.26) &  (0.24) &  (0.15) &  (0.14) &  (0.13) \\[0.05cm]
 &   &       & 0.2 &  -0.13 &    0.08 &    0.10 &   -0.08 &    0.06 &    0.04 &   -0.02 &    0.01 &    0.01 \\
 &   &       &     & (0.68) &  (0.59) &  (0.58) &  (0.38) &  (0.34) &  (0.29) &  (0.19) &  (0.17) &  (0.15) \\[0.1cm]
 &   &  0.4  & 0.1 &  -0.04 &    0.08 &    0.07 &   -0.02 &    0.02 &    0.02 &    0.00 &    0.01 &    0.01 \\
 &   &       &     & (0.34) &  (0.40) &  (0.38) &  (0.18) &  (0.21) &  (0.19) &  (0.10) &  (0.11) &  (0.11) \\[0.05cm]
 &   &       & 0.2 &  -0.03 &    0.09 &    0.08 &   -0.01 &    0.04 &    0.04 &    0.00 &    0.02 &    0.01 \\
 &   &       &     & (0.45) &  (0.55) &  (0.52) &  (0.25) &  (0.31) &  (0.29) &  (0.13) &  (0.15) &  (0.15) \\[0.2cm]
Continuous
 & 0 &  0.2  & 0.1 &  -0.19 &    0.10 &    0.12 &   -0.06 &    0.04 &    0.04 &   -0.02 &    0.01 &    0.01 \\
 \& Binary
 &   &       &     & (0.64) &  (0.56) &  (0.55) &  (0.30) &  (0.29) &  (0.26) &  (0.14) &  (0.14) &  (0.13) \\[0.05cm]
 &   &       & 0.2 &  -0.19 &    0.13 &    0.12 &   -0.07 &    0.05 &    0.04 &   -0.02 &    0.01 &    0.01 \\
 &   &       &     & (0.80) &  (0.75) &  (0.64) &  (0.38) &  (0.37) &  (0.31) &  (0.18) &  (0.18) &  (0.15) \\[0.1cm]
 &   &  0.4  & 0.1 &  -0.04 &    0.11 &    0.07 &   -0.01 &    0.03 &    0.02 &    0.00 &    0.01 &    0.00 \\
 &   &       &     & (0.34) &  (0.46) &  (0.37) &  (0.18) &  (0.24) &  (0.19) &  (0.10) &  (0.12) &  (0.10) \\[0.05cm]
 &   &       & 0.2 &  -0.04 &    0.13 &    0.12 &   -0.02 &    0.06 &    0.05 &    0.00 &    0.02 &    0.01 \\
 &   &       &     & (0.51) &  (0.70) &  (0.58) &  (0.26) &  (0.37) &  (0.29) &  (0.13) &  (0.18) &  (0.14) \\[0.05cm]
\hline
% &&\\[-0.2cm]
\end{tabular}
\end{footnotesize}
}
\end{table}

Theorem \ref{prop_tcl} establishes the asymptotic normality of our proposed estimator. However, as mentioned in Section \ref{sec:normal} of the main paper, the asymptotic covariance can be difficult to estimate in general, and, therefore, we suggest using bootstrapping. Table \ref{tab:rescov} shows the empirical coverage for 95\% confidence intervals constructed using bootstrapping with 399 replicates for the simulated data; we find that the empirical coverage is close to the nominal level. {\red The same is true for the cases with six covariates and binary covariates, but we do not present them here for the sake of brevity.}

\begin{table}[htbp]
\caption{Empirical coverage of 95\% bootstrapped confidence intervals \label{tab:rescov}}
\centering
%\smallskip
\begin{small}
\begin{tabular}{ccc@{~~~~}c@{~~}c@{~~}c@{~~~~}c@{~~}c@{~~}c@{~~~~}c@{~~}c@{~~}c}
\hline
&&& \multicolumn{3}{l}{\hspace{0.3cm}$n=100$} &  \multicolumn{3}{l}{\hspace{0.3cm}$n=300$} &  \multicolumn{3}{l}{\hspace{0.3cm}$n=1000$} \\
$\nu$ & $\pi_m$ & $\rho$ & $\theta_{(0)}$ & $\theta_{(1)}$ & $\theta_{(2)}$ & $\theta_{(0)}$ & $\theta_{(1)}$ & $\theta_{(2)}$ & $\theta_{(0)}$ & $\theta_{(1)}$ & $\theta_{(2)}$ \\[0.1cm]
\hline
&&&&&&&&&&&\\[-0.3cm]
0 &  0.2  & 0.1 & 91.2 & 93.2 & 93.9 & 93.5 & 93.2 & 94.3 & 94.0 & 94.8 & 94.0 \\
  &       & 0.2 & 94.7 & 95.0 & 94.8 & 94.8 & 94.5 & 93.2 & 93.8 & 93.0 & 93.5 \\[0.05cm]
  &  0.4  & 0.1 & 94.6 & 93.6 & 93.4 & 94.8 & 94.0 & 94.9 & 94.7 & 94.4 & 94.8 \\
  &       & 0.2 & 94.7 & 94.2 & 93.8 & 93.9 & 94.4 & 93.8 & 94.7 & 94.0 & 94.2 \\[0.1cm]
2 &  0.2  & 0.1 & 92.9 & 94.3 & 93.4 & 93.9 & 93.5 & 94.8 & 94.2 & 94.3 & 94.9 \\
  &       & 0.2 & 95.9 & 95.3 & 96.1 & 94.5 & 94.9 & 94.1 & 93.0 & 95.0 & 95.3 \\[0.05cm]
  &  0.4  & 0.1 & 93.3 & 93.0 & 92.4 & 93.9 & 94.6 & 94.6 & 95.0 & 93.8 & 95.0 \\
  &       & 0.2 & 94.7 & 95.3 & 94.1 & 94.0 & 94.4 & 95.3 & 94.2 & 94.5 & 94.9 \\[0.05cm]
%2 &  0.2  & 0.1 & 90.8 & 93.6 & 93.3 & 93.2 & 93.5 & 93.9 & 94.0 & 93.8 & 93.2 \\
%  &       & 0.2 & 94.8 & 96.1 & 95.5 & 94.6 & 94.1 & 94.0 & 94.5 & 94.2 & 94.0 \\[0.1cm]
%  &  0.4  & 0.1 & 94.1 & 93.2 & 92.6 & 94.2 & 92.6 & 93.2 & 94.6 & 94.0 & 94.3 \\
%  &       & 0.2 & 96.0 & 95.8 & 95.0 & 94.2 & 94.3 & 94.4 & 94.4 & 94.0 & 94.0 \\[0.15cm]
%  &  0.2  & 0.1 & 92.1 & 94.4 & 93.8 & 93.2 & 93.6 & 93.5 & 93.5 & 93.3 & 93.8 \\
%  &       & 0.2 & 94.5 & 95.3 & 95.8 & 93.5 & 95.2 & 93.3 & 94.0 & 94.0 & 93.8 \\[0.1cm]
%  &  0.4  & 0.1 & 94.5 & 94.0 & 94.0 & 94.2 & 92.9 & 93.2 & 94.0 & 94.5 & 94.1 \\
%  &       & 0.2 & 95.0 & 94.8 & 94.6 & 94.9 & 94.7 & 95.5 & 93.6 & 93.5 & 94.1 \\[0.1cm]
\hline
% &&\\[-0.2cm]
\end{tabular}
\end{small}
\end{table}

\begin{table}[htbp]
\caption{Correct zeros, incorrect zeros, and model degrees of freedom\label{tab:resvaradd}}
\centering
\begin{small}
\smallskip
\begin{tabular}{cccc@{~~~~}c@{~~}c@{~~}c@{~~~~}c@{~~}c@{~~}c@{~~~~}c@{~~}c@{~~}c}
\hline
&&&& \multicolumn{3}{l}{\hspace{0.3cm}$n=100$} &  \multicolumn{3}{l}{\hspace{0.3cm}$n=300$} &  \multicolumn{3}{l}{\hspace{0.3cm}$n=1000$} \\
Type & $\nu$ & $\pi_m$ & $\rho$ & C & IC & DF & C & IC & DF & C & IC & DF \\[0.1cm]
\hline
&&&&&&&&&&&&\\[-0.3cm]
oracle &   &       &     & 4.00 & 0.00 & 3.00 & 4.00 & 0.00 & 3.00 & 4.00 & 0.00 & 3.00\\[0.1cm]
lasso  & 0 &  0.2  & 0.1 & 2.86 & 0.37 & 3.77 & 2.39 & 0.00 & 4.61 & 2.31 & 0.00 & 4.69 \\[0.0cm]
       &   &       & 0.2 & 3.23 & 0.79 & 2.98 & 2.66 & 0.06 & 4.29 & 2.54 & 0.00 & 4.46 \\[0.05cm]
       &   &  0.4  & 0.1 & 2.52 & 0.11 & 4.37 & 2.33 & 0.00 & 4.67 & 2.30 & 0.00 & 4.70 \\[0.0cm]
       &   &       & 0.2 & 2.90 & 0.45 & 3.65 & 2.51 & 0.02 & 4.47 & 2.44 & 0.00 & 4.56 \\[0.1cm]
       & 2 &  0.2  & 0.1 & 2.76 & 0.34 & 3.90 & 2.40 & 0.00 & 4.60 & 2.28 & 0.00 & 4.72 \\[0.0cm]
       &   &       & 0.2 & 3.19 & 0.85 & 2.96 & 2.64 & 0.09 & 4.27 & 2.46 & 0.00 & 4.53 \\[0.05cm]
       &   &  0.4  & 0.1 & 2.58 & 0.10 & 4.32 & 2.37 & 0.00 & 4.63 & 2.24 & 0.00 & 4.76 \\[0.0cm]
       &   &       & 0.2 & 2.92 & 0.48 & 3.60 & 2.50 & 0.05 & 4.45 & 2.43 & 0.00 & 4.57 \\[0.1cm]
alasso & 0 &  0.2  & 0.1 & 3.44 & 0.37 & 3.19 & 3.47 & 0.01 & 3.52 & 3.67 & 0.00 & 3.33 \\[0.0cm]
       &   &       & 0.2 & 3.50 & 0.74 & 2.76 & 3.50 & 0.09 & 3.41 & 3.69 & 0.00 & 3.31 \\[0.05cm]
       &   &  0.4  & 0.1 & 3.34 & 0.17 & 3.49 & 3.52 & 0.00 & 3.48 & 3.69 & 0.00 & 3.31 \\[0.0cm]
       &   &       & 0.2 & 3.40 & 0.47 & 3.13 & 3.45 & 0.05 & 3.50 & 3.67 & 0.00 & 3.32 \\[0.05cm]
       & 2 &  0.2  & 0.1 & 3.35 & 0.36 & 3.30 & 3.48 & 0.01 & 3.51 & 3.65 & 0.00 & 3.35 \\[0.0cm]
       &   &       & 0.2 & 3.47 & 0.77 & 2.76 & 3.46 & 0.13 & 3.42 & 3.65 & 0.00 & 3.35 \\[0.05cm]
       &   &  0.4  & 0.1 & 3.33 & 0.15 & 3.52 & 3.55 & 0.00 & 3.44 & 3.72 & 0.00 & 3.28 \\[0.0cm]
       &   &       & 0.2 & 3.38 & 0.47 & 3.15 & 3.42 & 0.07 & 3.50 & 3.64 & 0.01 & 3.35 \\[0.05cm]
\hline
% &&\\[-0.2cm]
\end{tabular}
\end{small}
\end{table}

Table \ref{tab:resvaradd} displays the results of the adaptive lasso variable selection procedure for all simulation scenarios, as well as results for the lasso procedure. The IC values tend towards zero as the sample size increases for both the lasso and alasso. However, the lasso tends to select a more complex model than the alasso as indicated by the smaller C values and larger DF values; it is well known that the lasso exhibits this behaviour which is why alasso is preferred \citep{fanlv:2010}. Overall, the alasso works well with C approaching the oracle value of four as the sample size increases. The results are unaffected by the value of $\nu$ as we might expect (since $S_{T_0}$ is unspecified in our approach), whereas, when $n=100$, increased censoring proportion, $\rho$, or decreased cure proportion, $\pi_m$, both lead to fewer variables being selected.

{\red Table \ref{tab:Bsadd} displays the standard errors of the estimates (in the case with two continuous covariates) when estimation proceeds based on (i) knowing the true $S_C$ and using $B_i^*(S_C)$, (ii) estimating $S_C$ and using $B_i^*(\widehat S_C)$, and (iii) directly observing the cure labels, $B_i$. The theory of \citet{hitomi_nishiyama_okui_2008} predicts that the variability will be sequentially lower over these three cases, and this is clearly the case from Table \ref{tab:Bsadd}. Most importantly, there is efficiency \emph{gain} rather than loss by estimating $S_C$.}

\begin{table}[htbp]
\caption{{\red Standard error of estimates using different $B$ types}\label{tab:Bsadd}}
\centering
{\red
%\smallskip
\begin{footnotesize}
\begin{tabular}{cccc@{~~~~}c@{~~}c@{~~}c@{~~~~}c@{~~}c@{~~}c@{~~~~}c@{~~}c@{~~}c}
\hline
&&&& \multicolumn{3}{l}{\hspace{0.4cm}$n=100$} &  \multicolumn{3}{l}{\hspace{0.4cm}$n=300$} &  \multicolumn{3}{l}{\hspace{0.4cm}$n=1000$} \\
$\nu$ & $\pi_m$ & $\rho$ & $B$ Type & $\theta_{(0)}$ & $\theta_{(1)}$ & $\theta_{(2)}$ & $\theta_{(0)}$ & $\theta_{(1)}$ & $\theta_{(2)}$ & $\theta_{(0)}$ & $\theta_{(1)}$ & $\theta_{(2)}$ \\[0.1cm]
\hline
&&&&&&&&&&&&\\[-0.3cm]
 0 &  0.2  & 0.1 & $B^*(S_C)$      & (0.89) & (0.72) & (0.72) & (0.46) & (0.39) & (0.36) & (0.21) & (0.19) & (0.17) \\
   &       &     & $B^*(\widehat S_C)$ & (0.65) & (0.62) & (0.58) & (0.29) & (0.29) & (0.26) & (0.15) & (0.14) & (0.14) \\[0.03cm]
   &       &     & $B$             & (0.39) & (0.36) & (0.37) & (0.21) & (0.19) & (0.19) & (0.11) & (0.11) & (0.10) \\[0.08cm]
   &       & 0.2 & $B^*(S_C)$      & (1.06) & (0.88) & (0.92) & (0.74) & (0.60) & (0.58) & (0.38) & (0.31) & (0.28) \\
   &       &     & $B^*(\widehat S_C)$ & (0.80) & (0.74) & (0.68) & (0.40) & (0.40) & (0.34) & (0.20) & (0.20) & (0.17) \\[0.03cm]
   &       &     & $B$             & (0.38) & (0.37) & (0.36) & (0.21) & (0.20) & (0.19) & (0.11) & (0.10) & (0.10) \\[0.13cm]
   &  0.4  & 0.1 & $B^*(S_C)$      & (0.40) & (0.53) & (0.50) & (0.22) & (0.28) & (0.26) & (0.12) & (0.14) & (0.13) \\
   &       &     & $B^*(\widehat S_C)$ & (0.36) & (0.49) & (0.47) & (0.18) & (0.24) & (0.22) & (0.10) & (0.12) & (0.12) \\[0.03cm]
   &       &     & $B$             & (0.25) & (0.32) & (0.30) & (0.14) & (0.17) & (0.17) & (0.08) & (0.09) & (0.09) \\[0.08cm]
   &       & 0.2 & $B^*(S_C)$      & (0.68) & (0.84) & (0.85) & (0.39) & (0.50) & (0.44) & (0.19) & (0.24) & (0.21) \\
   &       &     & $B^*(\widehat S_C)$ & (0.51) & (0.74) & (0.66) & (0.29) & (0.37) & (0.32) & (0.13) & (0.18) & (0.16) \\[0.03cm]
   &       &     & $B$             & (0.26) & (0.31) & (0.30) & (0.14) & (0.17) & (0.17) & (0.08) & (0.09) & (0.09) \\[0.13cm]
 2 &  0.2  & 0.1 & $B^*(S_C)$      & (0.93) & (0.73) & (0.72) & (0.49) & (0.40) & (0.37) & (0.22) & (0.19) & (0.17) \\
   &       &     & $B^*(\widehat S_C)$ & (0.62) & (0.60) & (0.51) & (0.29) & (0.29) & (0.24) & (0.15) & (0.15) & (0.13) \\[0.03cm]
   &       &     & $B$             & (0.37) & (0.36) & (0.35) & (0.21) & (0.19) & (0.20) & (0.11) & (0.11) & (0.11) \\[0.08cm]
   &       & 0.2 & $B^*(S_C)$      & (1.12) & (0.86) & (0.87) & (0.73) & (0.59) & (0.53) & (0.41) & (0.33) & (0.28) \\
   &       &     & $B^*(\widehat S_C)$ & (0.80) & (0.76) & (0.60) & (0.45) & (0.43) & (0.32) & (0.21) & (0.20) & (0.15) \\[0.03cm]
   &       &     & $B$             & (0.38) & (0.35) & (0.36) & (0.21) & (0.20) & (0.20) & (0.11) & (0.10) & (0.11) \\[0.13cm]
   &  0.4  & 0.1 & $B^*(S_C)$      & (0.45) & (0.58) & (0.50) & (0.23) & (0.29) & (0.25) & (0.12) & (0.14) & (0.12) \\
   &       &     & $B^*(\widehat S_C)$ & (0.37) & (0.53) & (0.44) & (0.19) & (0.25) & (0.21) & (0.10) & (0.13) & (0.11) \\[0.03cm]
   &       &     & $B$             & (0.26) & (0.30) & (0.30) & (0.14) & (0.17) & (0.17) & (0.08) & (0.09) & (0.09) \\[0.08cm]
   &       & 0.2 & $B^*(S_C)$      & (0.73) & (0.91) & (0.83) & (0.38) & (0.47) & (0.36) & (0.20) & (0.25) & (0.18) \\
   &       &     & $B^*(\widehat S_C)$ & (0.55) & (0.77) & (0.60) & (0.30) & (0.41) & (0.27) & (0.15) & (0.20) & (0.13) \\[0.03cm]
   &       &     & $B$             & (0.26) & (0.29) & (0.30) & (0.15) & (0.17) & (0.17) & (0.08) & (0.09) & (0.09) \\[0.08cm]
\hline
% &&\\[-0.2cm]
\end{tabular}
\end{footnotesize}
}
\end{table}

\end{document}